\documentclass[12pt]{article}
\pdfoutput=1
\usepackage{draft,tikz-cd}

\newtheorem{conj}{Conjecture}
\newtheorem{defn}{Definition}
\newtheorem{thm}{Theorem}
\newtheorem{lem}{Lemma}

\newcommand{\su}{\mathfrak{su}}
\newcommand{\SU}{\text{SU}}
\newcommand{\SL}{\text{SL}}
\newcommand{\SCI}{\text{SC}}
\renewcommand{\S}{\text{S}}
\newcommand{\M}{\text{M}}
\newcommand{\FM}{\text{FM}}
\newcommand{\FS}{\text{FS}}
\newcommand{\U}{\text{U}}

\usepackage{tikz,tikz-3dplot}
\usepackage{pgfplots}
\pgfplotsset{compat=1.17}
\usetikzlibrary{shapes,arrows,cd,chains,decorations.markings,decorations.pathmorphing,calc}
\tikzset{
	->-/.style args={#1rotate#2}{decoration={markings, mark=at position #1 with {\arrow[scale=1.5,rotate = #2 ]{stealth}}}, postaction={decorate}}
}
\tikzstyle{GraphNode}=[circle, draw=black, fill=black, inner sep=2pt, minimum size=5pt]
\tikzstyle{GraphEdge}=[black]
\pgfmathsetmacro{\gS}{1}

\begin{document}

\begin{titlepage}

\begin{center}

\title{On $\frac18$-BPS black holes and the chiral algebra of $\mathcal{N}=4$~SYM}

\author{Chi-Ming Chang,$^{a,b,c}$ Ying-Hsuan Lin,$^{d}$ and Jingxiang Wu$^{e}$}

\address{${}^a$Yau Mathematical Sciences Center (YMSC), Tsinghua University, Beijing 100084, China}

\address{${}^b$School of Natural Sciences, Institute for Advanced Study, Princeton, NJ 08540, USA}

\address{${}^c$Beijing Institute of Mathematical Sciences and Applications (BIMSA), \\ Beijing 101408, China}

\address{${}^d$Jefferson Physical Laboratory, Harvard University, Cambridge, MA 02138, USA}

\address{${}^e$Mathematical Institute, University of Oxford,
Andrew-Wiles Building,  Woodstock Road, Oxford, OX2 6GG, UK}

\email{cmchang@tsinghua.edu.cn,  yhlin@fas.harvard.edu, wuj1@maths.ox.ac.uk}

\end{center}

\vfill

\begin{abstract}
We investigate the existence of $\frac18$-BPS black hole microstates in the $\su(1,1|2)$ sector of Type IIB string theory on $\mathrm{AdS}_5 \times \mathrm{S}^5$. As will be explained, these states are in one-to-one correspondence with the Schur operators comprising the chiral algebra of $\mathcal{N}=4$ super-Yang-Mills, and a conjecture of Beem et al.\ implies that the Schur sector only contains graviton operators and hence $\frac18$-BPS black holes do not exist. We scrutinize this conjecture from multiple angles. 
Concerning the macroscopic counting, we rigorously prove that the flavored Schur index cannot exhibit black hole entropy growth, and provide numerical evidence that the flavored MacDonald index also does not exhibit such growth.
Next, we go beyond counting to examine the algebraic structure, beginning by presenting evidence for the well-definedness of the super-$\mathcal{W}$ algebra of Beem et al., then using modular differential equations to argue for an upper bound on the lightest non-graviton operator if existent,
and finally performing a systematic construction of cohomologies to recover only gravitons.
Along the way, we clarify key aspects of the 4d/2d correspondence using the formalism of the holomorphic topological twist.

\end{abstract}

\vfill

\end{titlepage}

\tableofcontents

\section{Introduction}

Black holes are central to unraveling the enigmas in quantum gravity.
It is highly desirable to identify the ``simplest'' black holes whose properties and dynamics can be computed exactly. The anti-de Sitter/conformal field theory (AdS/CFT) correspondence \cite{Maldacena:1997re} gives a way to map this problem onto the boundary CFT, where AdS black holes correspond to ensembles of states/operators. Finding the simplest black holes becomes finding the simplest ensemble of CFT operators that resemble black holes.

In 4d $\cN \ge 2$ superconformal field theories, 
there exists a simple sector known as the Schur sector \cite{Beem:2013sza}, whose operators are in one-to-one correspondence with the currents of a 2d chiral algebra. The correlation functions of Schur operators with the kinematics restricted to certain supercharge cohomologies are meromorphic functions fully determined by the chiral algebra. Ensembles of Schur operators are potentially dual to the simplest black holes.

The existing literature does not contain much discussion about black holes in the Schur sector, and the available studies present 
conflicting evidence for their existence.
The Schur operators preserve two supercharges and are thus $\frac14$-BPS in ${\cal N}=2$ theories and $\frac18$-BPS in ${\cal N}=4$ theories. Note that there are three different $\frac18$-BPS sectors of $\cN = 4$, as reviewed in Appendix~\ref{Sec:BPS}, and this paper exclusively considers the so-called $\mathfrak{su}(1,1|2)$ sector.
The $\frac18$-BPS states can be counted (with signs) by the MacDonald index \cite{Gadde:2011uv}, which is defined as a specific limit of the superconformal index \cite{Kinney:2005ej}. The latter enumerates BPS operators preserving a single supercharge, and the limit imposes an additional BPS condition preserving a second supercharge. In the ${\cal N}=4$ super-Yang-Mills (SYM) theory with gauge group $\SU(N)$ or $\U(N)$, it has been suggested \cite{Choi:2018hmj} based on saddle point analyses that the (flavored) MacDonald index evaluated in the Cardy limit produces a large free energy. After a Legendre transform, it gives an order $N^2$ entropy that is enough to account for macroscopically-sized AdS black holes. However, to date, the literature has not reported any $\frac18$-BPS black hole solutions of supergravity.  Recently, \cite{Dias:2022eyq} constructed black hole solutions carrying the same two charges as the $\frac18$-BPS sector without enforcing the BPS condition, but the BPS limits of these non-BPS solutions are singular \cite{Prahar}.

Let us compare our present context with the case involving one fewer supercharge, namely $\frac{1}{16}$-BPS, for which there are significant results in the literature.
Explicit AdS black hole solutions that preserve two supercharges (one supercharge and one superconformal charge from the CFT point of view) were found in \cite{Gutowski:2004ez,Gutowski:2004yv,Chong:2005hr,Chong:2005da,Kunduri:2006ek}. More recently, the entropy of these solutions was reproduced by counting the $\frac{1}{16}$-BPS operators in the ${\cal N}=4$ SYM using the superconformal index in the large $N$ limit \cite{Cabo-Bizet:2018ehj,Choi:2018hmj,Benini:2018ywd}. Going beyond the index, one could directly study the $\frac{1}{16}$-BPS operators by mapping them to the cohomologies of a supercharge $Q$ \cite{Grant:2008sk,Chang:2013fba,Chang:2022mjp}. 
Cohomologies representing genuine black hole microstates remained elusive for more than a decade, the difficulty of which is that at low energies,  
the vast majority of BPS operators are 
graviton cohomologies, which in the large $N$ limit are dual to 10d (multi-)graviton states in AdS$_5 \times \S^5$.
The first non-graviton cohomology was finally discovered in \cite{Chang:2022mjp} whose intricate properties and generalizations were further explored in \cite{Choi:2022caq,Choi:2023znd,Budzik:2023xbr,Chang:2023zqk,Budzik:2023vtr}.

In this paper, we investigate 
the existence of $\frac18$-BPS black holes from the ${\cal N}=4$ SYM point of view, leveraging constraints from the chiral algebra and new understandings of indices. Based on the explicit construction of generators in the chiral algebra at low ranks, a conjecture on the form of all generators of the chiral algebra was proposed by Beem et al. \cite{Beem:2013sza}. We observe that the normal-ordered products of these generators and their superconformal descendants represent the entire graviton cohomology restricted to the Schur sector (see Section~\ref{sec:conj_BH}). Consequently, their conjecture can be rephrased as
\begin{conj}
    \label{Conj}
The space of Schur operators is isomorphic to the Schur graviton cohomology. 
\end{conj}
\noindent In particular, it forbids the existence of $\frac18$-BPS black holes. Given the profoundness of this implication, we will scrutinize this conjecture from multiple angles.

This conjecture directly contradicts the aforementioned result on the flavored MacDonald index in the Cardy limit \cite{Choi:2018hmj}. The order $N^2$ entropy implies that the number of operators with order $N^2$ energies in the Schur sector must grow as $\exp({N^2})$.
However, the growth of the graviton states in the $\tfrac18$-BPS Schur sector is bounded by the growth $\exp({N^\frac{5}{3}})$ of $\tfrac{1}{16}$-BPS graviton states \cite{Kinney:2005ej}. To investigate this puzzle, we first consider a one-fugacity specialization of the flavored MacDonald index known as the flavored Schur (FS) index. Using the closed-form expressions \cite{Hatsuda:2022xdv,Bourdier:2015wda}, we prove rigorous upper bounds on the microcanonical flavored Schur index (see Section~\ref{sec:bound_on_flavor_schur}):
\ie
d_{\FS}<N^{1+\varepsilon}\quad\text{for all}\quad N > N_0.
\fe
for any $\varepsilon>0$ and some large enough $N_0>0$. Furthermore, in Section~\ref{sec:computation_index}, we explicitly compute the microcanonical flavored MacDonald index and find evidence that its growth is similar to that of the Schur index and much slower than the superconformal index. All this points to the opposite of \cite{Choi:2018hmj}, and supports Conjecture~\ref{Conj}. In Section~\ref{Sec:Cardy}, we revisit the saddle point analysis of \cite{Choi:2018hmj} and identify subtleties concerning the purported large $N$ saddle.

Even if the Schur sector does not host enough operators to support black hole entropy growth, it remains possible that it contains non-graviton operators invalidating Conjecture~\ref{Conj}. Assuming the contrary of the conjecture leads to the following three scenarios:
\begin{enumerate}
    \item Graviton generators do not close as a chiral algebra, and one needs to include non-graviton generators. In this case, there must be at least one non-graviton operator appearing in the operator product expansion (OPE) of two graviton generators. The dimension of the lightest non-graviton operator is thereby bounded by 
    \ie\label{eqn:no_sub_alg_bound}
    h \le N +1. 
    \fe
    Contrasting the $\cO(N)$ bound with the $\cO(N^2)$ growth of the mass of an AdS black hole, such non-graviton operators should have a bulk interpretation as other types of objects, e.g.\ giant gravitons. In Section~\ref{Sec:Closure}, we present evidence for the closure.

    \item Graviton generators close and form a chiral \emph{sub}algebra of the full chiral algebra of the ${\cal N}=4$ SYM. In this case, the non-graviton operators must assemble into modules of the smaller graviton chiral algebra.
    By analyzing the modular differential equations generated by the null states in the chiral algebra, we derive in Section~\ref{Sec:Bootstrap} an upper bound on the dimension of the lowest non-graviton operator,
    \ie\label{eqn:MDE_bound}
        h \le \lfloor \frac{(N+1)^2}{2} \rfloor - 1.
    \fe

    \item The conjectured isomorphism between the vector space of Schur operators and the vector space of the chiral algebra is broken by quantum corrections. In this case, the vector space of Schur operators might contain non-graviton operators that have no counterpart in the chiral algebra.
\end{enumerate}
In Section~\ref{sec:brute}, for low-rank gauge groups, a brute-force construction of Schur operators recovers only gravitons, lending further credence to Conjecture~\ref{Conj}.

The rest of this paper is organized as follows. Section~\ref{sec:Schur_sector} introduces the Schur sector in the ${\cal N}=4$ SYM, examines its relation to the chiral algebra and the $Q$-cohomology, and discusses Conjecture~\ref{Conj}
and its implications on black holes. Section~\ref{sec:holotwist} discusses two potential subtleties on the $Q$-cohomology and their resolutions. Section~\ref{Sec:Indices} presents various indices that count (with signs) the number of Schur operators, performs saddle point analyses on the microcanonical flavored MacDonald index in the large $N$ limit,
proves rigorous bounds on the microcanonical Schur and flavored Schur indices, and explicitly evaluates various indices. Section~\ref{Sec:OpAlg} presents evidence for the closure of the graviton algebra, analyzes the modular differential equations satisfied by the Schur index, and
presents the results from the brute-force computer search for non-graviton operators.
Section~\ref{Sec:Remarks} ends with some concluding remarks.

\section{Schur sector and black holes}
\label{sec:Schur_sector}

This section presents two perspectives on the Schur sector and chiral algebra of the $\cN=4$ SYM and discusses their implications on $\frac18$-BPS black holes.

\subsection{Two cohomological formulations of the Schur sector}

The Schur operators can be viewed as the $\frac18$-BPS operators for the supercharges $Q^4_-$ and $\overline Q_{\dot-3}$. Their commutators with their Hermitian conjugates (BPZ conjugates) $Q^{4\dagger}_-= S_4^-$ and $\overline Q^\dagger_{\dot-3} = \overline S^{\dot - 3}$ are
\ie\label{eqn:QQ_BPS_bounds}
2\{Q^4_-,S_4^-\}&=D-J_1-J_2-q_1-q_2-q_3\ge 0\,,
\\
2\{\overline Q_{\dot-3},\overline S^{\dot - 3}\}&=D-J_1+J_2-q_1-q_2+q_3\ge 0\,,
\\
\{Q^4_-,\overline Q_{\dot-3}\}&=\{Q^4_-,\overline S^{\dot - 3}\}=\{\overline Q_{\dot-3},S_4^-\}=0\,.
\fe
The Schur operators saturate the above BPS bounds, i.e.
\ie\label{eqn:BPS_cond}
D-J_1-J_2-q_1-q_2-q_3=0\,,\quad J_2+q_3=0\,.
\fe
In the following, we review the standard cohomological formulation of the Schur sector giving rise to a chiral algebra, followed by an alternative formulation that makes contact with the enumeration of $\frac{1}{16}$-BPS operators.

\paragraph{Chiral algebra as the $Q+S$-cohomolgy}

The first way to arrive at the BPS conditions \eqref{eqn:BPS_cond} is to consider the linear combinations \cite{Beem:2013sza}
\ie
{\mathbb Q}_1=Q^4_-+\overline S^{\dot-3}\,,\quad{\mathbb Q}_2=S_4^--\overline Q_{\dot-3}\,.
\fe
The anti-commutators between them and their hermitian conjugates are
\ie\label{eqn:bbQ_comm}
\{{\mathbb Q}_1,{\mathbb Q}_1^\dagger\}&=\{{\mathbb Q}_2,{\mathbb Q}_2^\dagger\}=D-J_1-q_1-q_2\ge 0\,,
\\
\{{\mathbb Q}_1,{\mathbb Q}_2\}&=-J_2-q_3\,.
\fe
The Schur operators saturate the BPS bound on the first line of \eqref{eqn:bbQ_comm}. Note that once this BPS bound is saturated, the combination $J_2+q_3$ automatically vanishes due to the BPS bounds that will be derived later in \eqref{eqn:QQ_BPS_bounds}. 
Hence, by the standard Hodge theory argument, the space of Schur operators is isomorphic to each of the ${\mathbb Q}_1$- and ${\mathbb Q}_2$-cohomologies of local operators inserted at the origin. 

Local operators away from the origin can be obtained by translating Schur operators from the origin. However, $\mathbb{Q}_{1,2}$ do not commute with general translations, but only commute with the (twisted) translations along a complex plane $\mathbb{C}\subset \mathbb{R}^4$. Therefore, local operators away from the origin cannot be in the $\mathbb{Q}_i$-cohomology unless they reside on this complex plane. In addition, the (twisted) anti-holomorphic translations turn out to be $\mathbb{Q}_i$-exact, and hence operators only depend on the insertion point holomorphically. The holomorphic OPE endows the vector space of Schur operators with the structure of a vertex operator algebra (VOA). The above procedure associates any 4d $\mathcal{N}=2$ superconformal field theory with a VOA~\cite{Beem:2013sza}. In this paper, we will focus on the $\mathrm{VOA}\left[\mathrm{SYM}_{\mathfrak{g}}\right]$ associated with the $\mathcal{N}=4$ SYM with gauge 
algebra $\mathfrak{g}$, putting special focus on $\mathfrak{g}=\mathfrak{sl}_N$.\footnote{In this paper, we are not concerned with the global structure of the gauge group.} 

\paragraph{Equivalence to the $Q$-cohomology}
\label{Sec:Q}

Given \eqref{eqn:QQ_BPS_bounds}, because the two supercharges $Q^4_-$ and $\overline Q_{\dot-3}$ anti-commute with each other, one can consider the $(Q^4_- + \overline Q_{\dot-3})$-cohomology, which can be computed using the technique of spectral sequence.
This cohomology is isomorphic as a vector space to the space of Schur operators by the standard Hodge theory argument.

In the weak coupling limit, the operators in ${\cal N}=4$ SYM are constructed from gauge invariant combinations (words) of the fundamental fields and covariant derivatives (letters). The Schur letters---the letters satisfying \eqref{eqn:BPS_cond}---are 
\ie\label{eqn:Schur_letters}
\phi^i=\Phi^{4i}\quad{\rm for}~i=1,2\,,\quad \psi=-i\Psi_{+3}\,,\quad \lambda=\overline\Psi^4_{\dot+}\,,\quad D := D_{+\dot+}\,.
\fe
They can be assembled into a superfield $\Psi$ as
\ie\label{eqn:Schur_superfield}
\Psi(z,\theta_1,\theta_2)=-i\sum_{n=0}^\infty \frac{(zD)^n}{n!}\left[\frac{1}{n+1}z\lambda+2\theta_i\phi^i+2\theta_1\theta_2\psi\right]\,,
\fe
where $z=x^{+\dot+}$ is a complex coordinate, and $\theta_i$ are two fermionic auxiliary variables.
We will refer to the Schur words as the gauge invariant combinations of the Schur superfield $\Psi$. 
The set of Schur letters \eqref{eqn:Schur_letters} is a subset of the set of BPS letters, and the Schur superfield \eqref{eqn:Schur_superfield} is a restriction of the BPS superfield of \cite{Chang:2013fba}.

In the free limit, the Schur words always satisfy the BPS conditions \eqref{eqn:BPS_cond}, and the supercharges $Q^4_-$, $\overline Q_{\dot - 3}$ act trivially on the Schur words. When the gauge coupling $g_{\rm YM}$ is turned on, some Schur words are lifted and become non-BPS. As argued in \cite{Beem:2013sza}, according to the representation theory of the superconformal algebra, the lifted Schur words always furnish quartets of the form
\ie
({\cal O}, \quad Q^4_-{\cal O},\quad \overline Q_{\dot - 3}{\cal O},\quad   Q^4_-\overline Q_{\dot - 3}{\cal O})\,,
\fe
which trivialize in each of the $Q^4_-$ and $\overline Q_{\dot -3}$-cohomologies. The unlifted Schur words form singlets (annihilated by all four supercharges) of the algebra \eqref{eqn:QQ_BPS_bounds},
and contribute to both the $Q^4_-$ and $\overline Q_{\dot -3}$-cohomologies.
In conjunction with the previous Hodge theory argument, we know that the space of Schur operators is isomorphic to each of the $Q^4_-$ and $\overline Q_{\dot -3}$-cohomologies on the space of Schur words. In other words, the spectral sequence computing the $(Q^4_-+\overline Q_{\dot-3})$-cohomology degenerates on the first page.
For convenience, we choose to study the $Q^4_-$-cohomolgy. The supercharge $Q := Q^4_-$ acts on the superfield $\Psi$ as
\ie
Q\Psi = \Psi^2\,, \label{eq:QPsi}
\fe
which takes the same form as the supercharge acting on the BPS superfield in \cite{Chang:2013fba}.
Similar to the discussion in \cite{Chang:2013fba}, we recognize that the $Q$-cohomology is nothing but the relative Lie algebra cohomology
\begin{equation}
C^{\bullet}\left(\mathfrak{g}[[z, \theta^1,\theta^2]] \mid \mathfrak{g}\right).
\label{eq:liealgebracoh}
\end{equation}
Specifically, we will focus on the case of $\mathfrak{g}=\mathfrak{sl}_N$.

\subsection{Large $N$ cohomology and Schur graviton operators}

The Schur (single-)graviton operators are defined as single-trace words of the form
\ie\label{eqn:Schur_grav_op}
\partial^p_z\partial^{q_1}_{\theta_1}\partial^{q_2}_{\theta_2}\Tr\left[(\partial_z\Psi)^k (\partial_{\theta_1}\Psi)^{m_1}(\partial_{\theta_2}\Psi)^{m_2}\right]\big|_{z=0=\theta_i}\,,
\fe
where the letters inside the trace are assumed to be symmetrized. It is easy to see that \eqref{eqn:Schur_grav_op}
are $Q$-closed and represent non-trivial cohomology classes. Their products are also $Q$-closed but do not necessarily represent non-trivial cohomology classes due to the trace relations. As will be explained later, we refer to these cohomologies as the \emph{Schur graviton cohomologies}. 
\begin{defn}[Schur graviton cohomology]
    A Schur graviton cohomology is a $Q$-cohomology that has a representative given by a linear combination of products of the single-graviton operators \eqref{eqn:Schur_grav_op}.
\end{defn}
The Schur graviton cohomologies contain the Higgs and Hall-Littlewood chiral rings \cite{Gadde:2011uv}. In terms of the Schur superfield $\Psi$, the Higgs chiral ring is generated by
\ie\label{eqn:Higgs_chiral_ring}
\Tr\left[ (\partial_{\theta_1}\Psi)^{m_1}(\partial_{\theta_2}\Psi)^{m_2}\right]\big|_{z=0=\theta_i}\,,
\fe
whereas the Hall-Littlewood chiral ring is generated by
\ie
\Tr\left[(\partial_z\Psi)^k (\partial_{\theta_1}\Psi)^{m_1}(\partial_{\theta_2}\Psi)^{m_2}\right]\big|_{z=0=\theta_i}\,.
\fe

\subsection{Central conjecture, gravitons and black holes}\label{sec:conj_BH}

The Schur graviton operators play an important role in the chiral algebra description of the Schur sector. As conjectured in \cite{Beem:2013sza}, the chiral algebra of the ${\cal N}=4$ SYM is isomorphic to an ${\cal N}=4$ super ${\cal W}$-algebra with $N-1$ generators, which are the Higgs chiral ring generators \eqref{eqn:Higgs_chiral_ring} (counted up to $\mathfrak{su}(2)$ rotations)
of dimensions equal to half of the degree of the Casimirs of the gauge group. We see that these generators are all Schur graviton operators, and hence we will call them \emph{graviton generators}.
The Schur graviton cohomologies are isomorphic as a vector space to the vertex operators of the $\cN=4$ super $\cW$-algebra above.
In other words, the conjecture of \cite{Beem:2013sza} implies that on the space of Schur words, all the $Q$-cohomologies are Schur graviton cohomologies, and the conjecture can be reformulated as Conjecture~\ref{Conj}.

Let us compare the story here with that of the $\frac{1}{16}$-BPS operators studied in \cite{Chang:2022mjp,Choi:2022caq,Choi:2023znd,Chang:2023zqk,Budzik:2023vtr}. The space of $\frac{1}{16}$-BPS operators is isomorphic to the $Q$-cohomology on the space of BPS words, which are the gauge invariant combinations of (the derivatives of) BPS superfields. The BPS superfield is a function of two bosonic and three fermionic variables with an expansion similar to \eqref{eqn:Schur_superfield}, and a class of $Q$-cohomologies called graviton cohomologies is generated by the obvious generalization of the formula \eqref{eqn:Schur_grav_op}. In the large $N$ limit, the $\frac{1}{16}$-BPS operators dual to the graviton cohomologies correspond to the graviton states in ${\rm AdS}_5\times {\rm S}^5$, as evidenced by the matching between the partition functions \cite{Chang:2013fba}. At finite $N$, there exist $Q$-cohomologies that are not graviton cohomologies. They are called non-graviton cohomologies, or simply black hole cohomologies, as they are expected to capture the ${\rm AdS}$ black hole microstates. Their representatives are thereby called {\it black hole operators}.

With this understanding, we now know that Conjecture~\ref{Conj}, in particular, forbids any nontrivial Schur black hole cohomologies, and excludes the existence of $\frac18$-BPS black holes.

\section{Quantum corrections and collision ambiguities}\label{sec:holotwist}

At this point, one may worry about the following two issues:
\begin{enumerate}
    \item Quantum corrections to $Q$: As discussed above, once the interactions are turned on, the space of Schur operators jumps discontinuously, and we need to resort to the cohomology problem defined in \eqref{eq:QPsi}. This \emph{tree level} correction is precisely what is needed to match the VOA considerations in \cite{Beem:2013sza}. What about additional quantum corrections that generically seem to be present?
    \item Operator product: To construct composite BPS words out of BPS letters, one has to make sense of the colliding limit of operators.
    In free theory, the natural choice is to use normal ordering. In the presence of interactions, when the BPS letters are gauge invariant, we can use point-splitting regularization, i.e.\ remove the singular part of the OPE. However, when the BPS letters are not gauge invariant, like the case at hand, there is no obvious recipe.
\end{enumerate}

Both concerns can be remedied by the formalism of the holomorphic topological twist \cite{Kapustin:2006hi}, which puts the discussions in Section~\ref{sec:Schur_sector} on rigorous ground. More explicitly, it has been shown in \cite{Oh:2019bgz,Jeong:2019pzg,butsonEquivariantLocalizationFactorization} (see also \cite{Bobev:2020vhe}) that the chiral algebra can be realized by this twist in the $\Omega$ background. In the following subsections, we will review the formalism and discuss the resolutions of these two problems.

\subsection{Holomorphic topological twist}

The supercharge for the holomorphic topological (HT) twist \cite{Kapustin:2006hi} is  
\begin{equation}
    Q_{\mathrm{HT}} = Q^4_- +\overline Q_{\dot-3}.
\end{equation}
We define a one-form supercharge by
\begin{equation}
\mathbf{Q} := Q_{+}^3 \mathrm{d}x^{+\dot -} + \overline{Q}_{4\dot +} \mathrm{d} x^{-\dot +} + Q^3_{-}\mathrm{d}x^{-\dot -}.
\end{equation}
The translations $\partial_{+\dot -}$, $\partial_{-\dot +}$ and $\partial_{-\dot -}$ are $Q_{\mathrm{HT}}$-exact since 
\begin{equation}
    \{Q_{\mathrm{HT}}, \mathbf{Q}\} = P_{+\dot -} \mathrm{d}x^{+\dot -} +P_{-\dot +}\mathrm{d}x^{-\dot +} + P_{-\dot -} \mathrm{d} x^{-\dot -}.
\end{equation}
Passing to the $Q_{\mathrm{HT}}$-cohomology, we obtain a theory that is topological along the $x^{+\dot -}$ and $x^{-\dot +}$ directions and holomorphic along $z:= x^{+\dot +}$. As a consequence of Hartog's theorem, which says there can be no isolated singularities in $\mathbb{C}^2$, the OPEs in this HT-twisted theory are entirely regular, and the composite operators are simply constructed as juxtapositions of fundamental fields.\footnote{More precisely, the singularities in the OPE between $Q_{\rm HT}$-closed operators are $Q_{\rm HT}$-exact.}

The $Q_{\mathrm{HT}}$-cohomology can be readily described in the Batalin–Vilkovisky (BV) formalism. More precisely, we consider a twisted BRST charge $\mathcal{Q}$ given by adding the supercharge $Q_{\mathrm{HT}}$ to the original BRST charge.
The field resides in the space \cite{costelloLecturesMathematicalAspects2015,butsonEquivariantLocalizationFactorization}
\begin{equation}
    \mathbf{C} \in \Omega_{\mathbb{C}}^{0,\bullet} \otimes \Omega_{\mathbb{R}^2}^{\bullet} \otimes \mathfrak{g}[\theta_1,\theta_2] \label{eq:calQC}
\end{equation}
with the standard BV action
\begin{equation}
    \int dz d\theta_1d\theta_2 \ \Tr \left(\frac12 \mathbf{C} d\mathbf{C} + \frac{1}{6}\mathbf{C}[\mathbf{C},\mathbf{C}]\right), \label{eq:BVactioncalC}
\end{equation}
where $d$ is the sum of the de Rham differential along the topological plane $\mathbb{R}^2$ and the Dolbeault differential along $\mathbb{C}$.

The $Q_{\mathrm{HT}}$-cohomology of local operators, isomorphic to the space of Schur operators, is given by the cohomology of the BRST differential $\mathcal{Q}$. Its action at the classical level is simply
\begin{equation}
    \mathcal{Q} \mathbf{C}=d \mathbf{C}+\frac{1}{2}[\mathbf{C}, \mathbf{C}],
\end{equation}
but it will generically receive quantum corrections.
Let us define
\ie
Q_{0} := \mathcal{Q} - d.
\fe
The computation of the $\mathcal{Q}$-cohomology could be simplified through a spectral sequence. We first pass to the $d$-cohomology, where the only contribution comes from the zero-form component
of $\mathbf{C}$,\footnote{Intuitively, first passing to the $d$-cohomology is restricting to the space of Schur operators in the free theory. We always use the corresponding un-boldface letter to denote the zero-form component.
}
\begin{equation}
    C[z,x_i,\theta^1,\theta^2]  := c(z,x_i) + \theta^i q_i(z,x_i) + \theta^1\theta^2 b(z,x_i) \in \Omega_{\mathbb{C}}^{0,0} \otimes \Omega_{\mathbb{R}^2}^{0} \otimes \mathfrak{g}[\theta_1,\theta_2]. \label{eq:Csuperfield}
\end{equation}
The tree level piece\footnote{This piece is referred to as $1$-loop term in \cite{Beem:2013sza}.} $Q_0$ is the second term of \eqref{eq:calQC}, and acts on the second page as
\begin{equation}
    Q_{0} C = \frac{1}{2}[C,C],
    \label{eq:Q0intC}
\end{equation}
which in component form becomes
\begin{equation}
    Q_{0} c = \frac12 [c,c], \quad Q_{0} q_i = [c,q_i], \quad Q_{0}b = [c,b]  + \epsilon^{ij} [q_i, q_j].
\end{equation}
Acting on words, $Q_0$ satisfies the Leibniz rule. In particular, one can verify that $Q_0$ is nilpotent and exactly reproduces \eqref{eq:QPsi}. The $Q_0$-cohomology can be neatly phrased as the relative Lie algebra cohomology 
\begin{equation}
C^{\bullet}\left(\mathfrak{g}[[z, \theta^1,\theta^2]] \mid \mathfrak{g}\right).
\end{equation}

Corrections to the BRST differential induced by the interaction terms in \eqref{eq:BVactioncalC} generally have a perturbative expansion in $\hbar$, i.e.\ $Q_{\mathrm{int}} = Q_0 + \hbar Q_1 + \hbar^2 Q_2 +\dots$, where $Q_n$ denotes the contributions from $n$-loop diagrams and acts trivially on the words of length shorter than $n+1$. Its action on longer words is nontrivial and requires explicit Feynman diagram calculations. The presence of a nontrivial $Q_n$ action for $n>0$ means that the BRST differential $Q_{\mathrm{int}}$ violates the Leibniz rule, and the $Q_{\mathrm{int}}$-cohomology is \emph{not} a Lie algebra cohomology.

Surprisingly, we will show below that all the loop diagrams naively contributing to the BRST differential vanish identically, and therefore \eqref{eq:Q0intC} holds to all orders in perturbation theory. The proof is a simple adaptation of the proof presented in section~6.1 of \cite{Budzik:2022mpd}. It was shown that for any holomorphic topological theory with $T$ number of topological directions and $H$ number of holomorphic directions, the only Feymann diagrams that can contribute are the so-called $(H+T)$-Laman graphs (more details can be found in Appendix~\ref{app:laman}), and the Feynman integrals can be reduced to the integrals of top form $\omega_{\Gamma}^{H,T}$ 
on the positive real projective space of Schwinger parameters. For any given $(H+T)$ Laman graph $\Gamma$, $\omega_{\Gamma}^{H, T}$ is fully explicit. 
When $H=0$ and $T=2$, $\omega_{\Gamma}^{H,T}$
vanishes identically for $2$-Laman graphs, which renders all relevant loop diagrams vanishing, and thus provides an alternative proof of the Kontsevich formality theorem. The example at hand corresponds to $H =1$ and $T=2$, where we can explicitly check that $\omega_{\Gamma}^{H,T}$
vanishes for $3$-Laman graphs. 
As a consequence, $Q_{\mathrm{int}}$ truncates at tree level to $Q_0$ and satisfies the Leibniz rule.\footnote{The leading loop correction, i.e.\ $Q_1$ in the notation here, was computed in \cite{Zwiebel:2005er}. It was shown that $Q_1$ takes the form $[Q_0, -]$, suggesting that $Q_1$ is trivial under an appropriate choice of basis, and is consistent with our proof here.} The space of Schur operators is then given by the $Q_0$ cohomology \eqref{eq:liealgebracoh} to all loop orders.

\subsection{$\Omega$ deformation to chiral algebra}

The discussion so far does not provide a VOA structure since intuitively, operators can move freely in topological directions and all the OPEs are regular. To get singular OPEs, we need to get rid of the topological directions. This can be done by the $\Omega$ deformation \cite{Nekrasov:2010ka}; see also \cite{Costello:2016nkh} for a nice summary from a modern perspective. Denote the vector field that generates the rotation on $\mathbb{R}^2$ by $V$. We can turn on the $\Omega$ background associated with the rotation on the topological plane by deforming the BRST differential to
\begin{equation}
    Q_{\mathrm{HT}} +  \iota_V \mathbf{Q},
\end{equation}
which is not nilpotent but squares to Lie derivative for the rotation vector field $V$.

This localizes the 4d theory to a complex plane at the origin of the topological plane, whose space of fields is simply the restriction of the bulk fields $\mathbf{C}$, 
\begin{equation}
    \mathbf{C}(z,\theta^1,\theta^2) = \mathbf{c}(z) +\theta^i \mathbf{q}_i(z) +\theta^1\theta^2 \mathbf{b}(z) \in \Omega_{\mathbb{C}}^{0,\bullet}  \otimes \mathfrak{g}[\theta_1,\theta_2].
    \label{eq:calQC2d}
\end{equation}
With slight abuse, we denote the 2d restriction of the 4d fields by the same symbols, with the caveat that after restricting operators to the complex plane $x_i=0$, the 2d OPE has singularities.
The BV action of the 2d localized theory \cite{Oh:2019bgz,Jeong:2019pzg} takes the familiar form
\begin{equation} \int_{\mathbb{C}}dz\operatorname{Tr} \left( \mathbf{b} \bar{\partial} \mathbf{c}+ \epsilon^{ij} \mathbf{q}_i \bar{\partial} \mathbf{q}_j + \epsilon^{ij} \mathbf{c}[\mathbf{q}_i,\mathbf{q}_j] + \mathbf{b}[\mathbf{c},\mathbf{c}] \right).
\end{equation}

The 4d $Q_0$-cohomology problem defined by \eqref{eq:Csuperfield} and \eqref{eq:Q0intC} descend to 2d straightforwardly. The $\bar{\partial}$-cohomology on the first page restricts to the zero-form components of $\mathbf{C}(z,\theta^1,\theta^2)$ that are holomorphic in $z$,
\begin{equation}
    C[z,\theta^1,\theta^2]  := c(z) + \theta^i q_i(z) + \theta^1\theta^2 b(z) \in \Omega_{\mathbb{C}}^{0,0} \otimes  \mathfrak{g}[\theta_1,\theta_2].
\end{equation}

The tree level BRST differential $Q_0^{2d}$ is given by a single Wick contraction with the interaction term $J_{\mathrm{BRST}} = \Tr b [c,c]+ \epsilon^{ij} c[q_i, {q}_j]$, as shown in Figure~\ref{fig:Wick}, 
\begin{figure}[htb]
    \centering
    \includegraphics[width=0.45\textwidth]{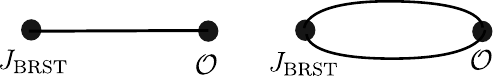}
    \caption{Wick contractions for $Q_0^{2d}\mathcal{O}$ and $Q_1^{2d}\mathcal{O}$}
    \label{fig:Wick}
\end{figure}
yielding exactly the same transformation as in \eqref{eq:Q0intC},
\begin{equation}
    Q_0^{2d} C = \frac12 [C,C]. \label{eq:Q02dC}
\end{equation}

However, unlike the holomorphic topological twist in 4d, there is no non-renormalization theorem in 2d and indeed the BRST differential receives quantum corrections. We denote by $Q_1^{2d}$ the contribution from two Wick contractions
\begin{align}
    Q_1^{2d} \Big(q^A_k q^B_l\Big) & = 2\epsilon_{kl}f^{ABC} \partial c^C, &  Q_1^{2d} \Big(b^A b^B\Big) & = -2 f^{ABC} \partial b^C, \\
    Q_1^{2d} \Big(q^A_k b^B\Big) & = 2 f^{ABC} \partial q^C_k, & Q_1^{2d} \Big(b^A c^B\Big) & = -2 f^{ABC} \partial c^C.
\end{align}
where the capital letters denote the gauge indices and $f^{ABC}$ are the structure constants. The $Q_1^{2d}$ action can be neatly packaged into
\begin{equation}
    Q_1^{2d}C^A(z,\theta)C^B(z,\theta') = -2 (\theta^1-\theta'^1)(\theta^2-\theta'^2) f^{ABC} \partial C^C.
\end{equation}
At the next order, the action of $Q_2^{2d}$ is trivial since it requires Wick-contracting all three fields in $J_{\mathrm{BRST}}$, and returns derivatives of the identity operators, which of course vanish. There are no further corrections.
Furthermore, the action of $Q_0^{2d}$ and $Q_1^{2d}$ can be combined into the standard contour integration of the BRST current
\begin{equation}
Q_{\mathrm{BRST}}^{2d}=\frac{1}{\hbar} \oint \frac{d z}{2 \pi i}\Tr\left( b [c,c]+ \epsilon^{ij} c[q_i, {q}_j]\right).
\end{equation}
This justifies the conjecture in \cite{Beem:2013sza} that $\mathrm{VOA}[\mathrm{SYM}_{\mathfrak{g}}]$ is the BRST reduction of the adjoint valued symplectic boson with the zero modes of $c$ removed.

There remains one key piece of the puzzle: Why do $Q_0^{4d}$ and $Q_{\mathrm{BRST}}^{2d}\equiv Q_0^{2d} + Q_1^{2d}$ have isomorphic cohomologies? We leave this issue for a subsequent paper, but let us spell out the problem a little bit in the remainder of this section.

The two differentials $Q_0^{4d}$ and $Q_{\mathrm{BRST}}^{2d}$ act on single-letters (derivatives of the free fields $\{b,c,q_i\}$) in the same way, but they act on composite operators very differently. To properly define composite operators, we need to specify the ways of regularizing the OPE singularities. There are two natural ways of doing this, each manifesting the 2d or 4d perspectives. The first way is to define composite operators by using the 4d OPE in the holomorphic topological twist, which are \emph{regular} thanks to Hartog's lemma. This descends to the free theory normal ordered product in the chiral algebra defined by subtracting off all possible Wick contractions. Let us denote the normal ordered product as
\ie\label{eqn:normal_order_product}
:f_1f_2\cdots f_n:\,,
\fe
where $f_r$ are derivatives of free fields $\{b,c,q_i\}$. The normal ordered product is associative and graded commutative. We use this definition for the composite operators throughout this paper unless stated otherwise.

The second way of defining composite operators is more standard in the VOA literature. We define the regularized product $(O_1O_2)$ as the first regular term in the 2d OPE of $O_1$ and $O_2$,\footnote{See also \eqref{eqn:reg_2_product} for a contour integral definition for the regularized product.} and the nested regularized product is defined by
\begin{equation}\label{eqn:regularized_product}
    (O_1O_2 \dots O_n)\ \equiv \ (O_1(O_2(O_3\dots (O_{n-1}O_n))))\,.
\end{equation}
The nested regularized product is not associative or commutative in general.
When the operators are derivatives of free fields, the normal order product is equal to the nested regularized product,
\ie
:f_1f_2\cdots f_n: \ = \ (f_1f_2\cdots f_n)\,.
\fe

$Q_0^{4d}$ satisfies the Leibniz rule when acting on the normal ordered product \eqref{eqn:normal_order_product}, while $Q_{\mathrm{BRST}}^{2d}$ satisfies the Leibniz rule when acting on the regularized product \eqref{eqn:regularized_product} following from the identity \eqref{eq:app:derivation}. 

\section{Indices}
\label{Sec:Indices}

Consider the thermal partition function
\ie
Z=\Tr\Omega\,,\quad \Omega=e^{-2\beta \{Q^4_-,S_4^-\}-\omega_1 J_1-\omega_2 J_2-\Delta_1q_1-\Delta_2q_2-\Delta_3q_3}\,.
\fe
The chemical potentials $\Delta_i$ and $\omega_i$ have $4\pi i$ periodicity due to the half-integral quantization of charges. The operator $\Omega$ anti-commutes with the supercharge $Q^4_-$ if the chemical potentials satisfy the relation
\ie\label{eqn:linear_chem_rel}
\Delta_1+\Delta_2+\Delta_3-\omega_1-\omega_2 \equiv 2\pi i\mod 4\pi i\,.
\fe
The superconformal (SC) index is defined to be the thermal partition function with this linear relation imposed,
\ie
I_\SCI=Z\big|_{\eqref{eqn:linear_chem_rel}}\,.
\fe
It is independent of the inverse temperature $\beta$ by standard arguments and can be written as
\ie
I_\SCI&=\Tr\left[(-1)^{F} e^{-\omega_1 (J_1-J_2)-\Delta_1(q_1+J_2)-\Delta_2(q_2+J_2)-\Delta_3(q_3+J_2)}\right]\,,
\fe
where we used the spin-statistics relation to write $(-1)^F=(-1)^{2J_2}$.
The states that contribute to the superconformal index must satisfy the first BPS condition in \eqref{eqn:BPS_cond}. Note that the chemical potentials $\Delta_i$ and $\omega_1$ in the superconformal index now have $2\pi i$ periodicity.

The {\it flavored MacDonald} (FM) index is given by the superconformal index in the limit
\ie\label{eqn:mac_limit}
\Delta_3\to \infty\,,\quad \Delta_1,\,\Delta_2,\,\omega_1~\text{fixed}\,,
\fe
which projects the trace to the subspace satisfying the second BPS condition in \eqref{eqn:BPS_cond}. We have
\ie
I_\FM& = \Tr_{J_2+q_3=0}\left[ (-1)^F e^{-\omega_1 (J_1-J_2)-\Delta_1(q_1+J_2)-\Delta_2(q_2+J_2)}\right]\,.
\fe
The {\it flavored Schur} (FS) index is the MacDonald index under the additional condition
\ie
\Delta_1+\Delta_2-\omega_1 \equiv 0\mod 2\pi i\,.
\fe
The (unflavored) {\it MacDonald} (M) and {\it Schur} (S) indices are given by further imposing the condition
\ie
\Delta_1 \equiv \Delta_2 \mod 2\pi i
\fe
on the corresponding flavored indices. For convenience, we solve the linear relation \eqref{eqn:linear_chem_rel} by parameterizing the chemical potentials with new fugacities $x, s, b, p$ such that
\ie
e^{-\Delta_1}=bxs\,,\quad e^{-\Delta_2}=b^{-1}xs\,,\quad e^{-\Delta_3}=p s^{-2}\,,\quad e^{-\omega_1}=x^2\,,\quad e^{-\omega_2}=p\,.
\fe
The various indices above expressed in these new fugacities become
\ie
I_\SCI&=\Tr\left[(-1)^F x^{2J_1+q_1+q_2} s^{q_1+q_2-2q_3}b^{q_1-q_2}p^{q_3+J_2}\right]\,,
\\
I_\FM&=\Tr_{q_3+J_2=0}\left[(-1)^F x^{2J_1+q_1+q_2} s^{q_1+q_2-2q_3}b^{q_1-q_2}\right]\,,
\\
I_\M&=\Tr_{q_3+J_2=0}\left[(-1)^F x^{2J_1+q_1+q_2} s^{q_1+q_2-2q_3}\right]\,,
\\
I_\FS&=\Tr_{q_3+J_2=0}\left[(-1)^F x^{2J_1+q_1+q_2} b^{q_1-q_2}\right]\,,
\\
I_\S&=\Tr_{q_3+J_2=0}\left[(-1)^F x^{2J_1+q_1+q_2}\right]\,.
\fe
The 
{\it microcanonical} indices are defined as the expansion coefficients of the corresponding indices. They are functions of the angular momenta $J_i$ and R-charges $q_i$.  We denote them by $d_\SCI$, $d_\FM$, $d_\M$, $d_\FS$, and $d_\S$.

The indices can be computed by unitary matrix integrals \cite{Kinney:2005ej},
\ie\label{Integral}
I&=\int dU\,\exp\left[\sum_{n=1}^\infty \frac{1}{n}\iota({\bf q}^n) \tr(U^n)\tr(U^{-n})\right]
\\
&=\frac{1}{N!}\int^\frac12_{-\frac12}\prod^N_{a=1}du_a\prod_{a<b}\left(2\sin \pi u_{ab}\right)^2\exp\left[\sum_{n=1}^\infty \frac{1}{n}\iota({\bf q}^n)\sum_{a,b=1}^N e^{2\pi i n u_{ab}}\right]
\\
&=\frac{1}{N!}\int^\frac12_{-\frac12}\prod^N_{a=1}du_a\exp\bigg[N\sum_{n=1}^\infty\frac{\iota({\bf q}^n)}{n}-\sum_{n=1}^\infty \frac{1}{n}\left[1-\iota({\bf q}^n)\right]\sum_{\substack{a,b=1\\a\neq b}}^N e^{2\pi i n u_{ab}}\bigg]\,,
\fe
where in the first line the integral is over the unitary matrices, and in the second line the off-diagonal matrix components are integrated out, leaving us with the chemical potentials $u_a$ corresponding to the Cartan of the U(N) gauge group. The function $\iota({\bf q})$ is the single-letter index, where ${\bf q}$ collectively denotes all the fugacities. The single-letter superconformal index and flavored MacDonald index are 
\ie
\iota_\SCI &=1- \frac{(1-bxs)(1-b^{-1}xs)(1-ps^{-2})}{(1-x^2)(1-p)}\,,
\\
\iota_\FM &=1- \frac{(1-bxs)(1-b^{-1}xs)}{(1-x^2)}\,,
\fe
and restrictions of the latter give the other single-letter indices.

\subsection{Flavored MacDonald index in the Cardy limit}
\label{Sec:Cardy}

Let us consider the Cardy limit of the flavored MacDonald index. The integral formula for the flavored MacDonald index is
\ie
I_\FM&\sim\frac{1}{N!}\int^\frac12_{-\frac12}\prod^N_{a=1}du_a\exp\bigg[-\sum_{n=1}^\infty \frac{1}{n}\frac{(1-e^{-n\Delta_1})(1-e^{-n\Delta_2})}{1-e^{-n\omega_1}}\sum_{\substack{a,b=1\\a\neq b}}^N e^{2\pi i n u_{ab}}\bigg]\,,
\fe
where we only keep the order $N^2$ term in the exponent of the integrated in \eqref{Integral}.

In the Cardy limit $\omega_1\to 0$, the infinite sum over $n$ can be computed exactly,
\ie
I_\FM&\sim\frac{1}{N!}\int^\frac12_{-\frac12}\prod^N_{a=1}du_a\exp\bigg[-\frac{1}{\omega_1}\sum_{\substack{a,b=1\\a> b}}^N V(u_{ab})\bigg]\,,
\fe
where
\ie
V(u)=\sum_{s_1,s_2=0}^1(-1)^{s_1 +s_2}\left[{\rm Li}_2(e^{2\pi i u-s_1 \Delta_1 - s_2\Delta_2})+{\rm Li}_2(e^{-2\pi i u-s_1 \Delta_1 - s_2\Delta_2})\right]\,.
\fe
In \cite{Choi:2018hmj}, the authors assumed the existence of a saddle point at 
\ie\label{eqn:not_valid_saddle}
u_1\approx u_2\approx\cdots\approx u_N\,.
\fe
However, the potential $V(u)$ is not smooth at $u=0$,
as we demonstrate by plotting $V(u)$ at $\Delta_1=\Delta_2=\frac{\pi}{2}(1+i)$ in Figure~\ref{Fig:cardy}. The non-smoothness at $u=0$ casts doubt on the saddle point \eqref{eqn:not_valid_saddle}. It is possible that the higher order terms in the small $\omega_1$ expansion smooth out the kink at $u=0$ and remedy the saddle point, but this certainly requires a careful analysis, which we presently do not pursue.\footnote{We thank Sunjin Choi and Seok Kim for correspondence regarding this point.} 

\begin{figure}[H]
	\centering
	\includegraphics[width=0.45\textwidth]{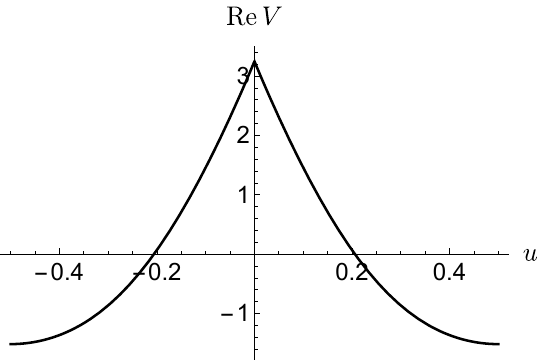} \qquad
	\includegraphics[width=0.45\textwidth]{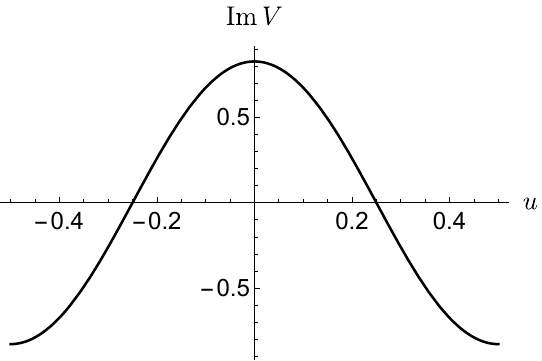}
	\caption{The real and imaginary parts of the potential $V(u)$ at $\Delta_1=\Delta_2=\frac{\pi}{2}(1+i)$.}
	\label{Fig:cardy}
\end{figure}

\subsection{Canonical and microcanonical saddles}

In the large $N$ limit, the integral formula \eqref{Integral} for the superconformal index can be evaluated using a saddle point approximation \cite{Cabo-Bizet:2019eaf,Cabo-Bizet:2020nkr,Cabo-Bizet:2020ewf,Choi:2021rxi}.
For unrestricted fugacities, a series of saddle points were found in \cite{Choi:2021rxi}, with the dominant contribution to the superconformal index being
\ie\label{eqn:large_N_index}
\log I_\SCI = \frac{N^2\Delta_1\Delta_2\Delta_3}{2\omega_1\omega_2}\Big|_{\eqref{eqn:linear_chem_rel}}\,.
\fe
The microcanonical superconformal index $d_\SCI$ is obtained by an inverse Laplace transform, which could be performed by yet another saddle point approximation. The result is given by the maximization
\ie
S=\log d_\SCI=\max\left(\frac{N^2\Delta_1\Delta_2\Delta_3}{2\omega_1\omega_2}+\Delta_1 q_1+\Delta_2 q_2+\Delta_3 q_3+\omega_1 J_1+\omega_2 J_2\Big|_{\eqref{eqn:linear_chem_rel}}\right)\,,
\fe
and was found to agree with the entropy of $\frac{1}{16}$-BPS black holes \cite{Hosseini:2017mds,Cabo-Bizet:2018ehj,Choi:2018hmj}.

Suppose the saddle point giving \eqref{eqn:large_N_index} continues to dominate in the MacDonald limit \eqref{eqn:mac_limit}, then we find
\ie
\log I_\FM = \frac{N^2\Delta_1\Delta_2}{2\omega_1}\,.
\fe
We compute the microcanonical MacDonald index by taking the inverse Laplace transform,
\ie
d_\FM = \frac{1}{(2\pi i)^3}\int d\Delta_1 d\Delta_2 d\omega_1 \,\exp\left(\frac{N^2\Delta_1\Delta_2}{2\omega_1} +\Delta_1 q_1+\Delta_2 q_2 +\omega_1 J_1\right)\,,
\fe
where the $\Delta_1,\,\Delta_2,\,\omega_1$ integrals are along the line $\gamma + i\bR$ with $\gamma>0$. The integrand has a flat direction parametrized by the variable $2\pi ix=\Delta_1+\Delta_2-\omega_1$. We compute the integral along the non-flat directions by saddle point approximation and perform the integral along the flat direction directly to find
\ie
d_\FM &\sim \int_{-\infty}^\infty dx\,\Big(e^{\frac{\pi i}{2}x \left(N^2+2q_1+2q_2+\sqrt{8N^2 J_1+N^4+4(q_1-q_2)^2+4N^2(q_1+q_2)}\right)}
\\
&\qquad+e^{\frac{\pi i}{2}x\left(N^2+2q_1+2q_2-\sqrt{8N^2 J_1+N^4+4(q_1-q_2)^2+4N^2(q_1+q_2)}\right)}\Big)
\sim \delta(J_1N^2-2q_1q_2)\,,
\fe
where we have used the inequality
\ie
2J_1+q_1+q_2\ge 0
\fe
valid for BPS operators. We see that the microcanonical flavored MacDonald index does not exhibit the order $\exp(N^2)$ black hole entropy growth.

\subsection{Rigorous asymptotics}
\label{sec:bound_on_flavor_schur}

It is possible to rigorously prove the absence of black hole entropy growth in some of the indices leveraging alternative closed-form formulae.  To study asymptotics, the difference between $\U(N)$ and $\SU(N)$ gauge groups is inconsequential, and we will consider $\U(N)$ below.

We present several definitions and lemmas that will facilitate the statements and proofs of this section.

\begin{defn}\label{Def:PartitionSeries}
    A partition series $\cA$ with one-sided fugacities $x_1, \dotsc, x_n$ and two-sided fugacities $y_1, \dotsc, y_m$ is an element
    \[
        \cA \in \bZ[y_1, \dotsc, y_m, y_1^{-1}, \dotsc, y_m^{-1}][[x_1, \dotsc, x_n]]
    \]  
    with
    \[
        \cA(x_1, \dotsc, x_n; y_1, \dotsc, y_m) = \sum_{p_1, \dotsc, p_n \ge 0; q_1, \dotsc, q_m} a(p_1, \dotsc, p_n; q_1, \dotsc, q_m) x_1^{p_1} \dotsb x_n^{p_n} y_1^{q_1} \dotsb y_m^{q_m},
    \]
    \[
        A(P_1, \dotsb, P_n) = \sum_{p_1 \le P_1, \dotsc, p_n \le P_n; q_1, \dotsc, q_m} |a(p_1, \dotsc, p_n; q_1, \dotsc, q_m)|.
    \]
    We call $a$ the coefficients (signed degeneracies) of $\cA$ and call $A$ the cumulative absolute degeneracies of $\cA$.  We further refer to $(\cA, a, A)$ as a partition triple.
\end{defn}

\begin{defn}\label{Def:LogLinear}
    A partition series
    \[ 
        \cA \in \bZ[y_1, \dotsc, y_m, y_1^{-1}, \dotsc, y_m^{-1}][[\mu, x]]
    \]
    is said to be log-linearly bounded if the cumulative absolute degeneracies $A$ are such that for any pair $\alpha, \varepsilon > 0$, there exists $N_0 \in \bZ_{>0}$ such that
    \[
        \log A(N, \lfloor\alpha N^2\rfloor) < N^{1+\varepsilon} \quad\text{for all}\quad N > N_0.
    \]
    Likewise, a family of partition series
    \[ 
        \{ \cA_N \in \bZ[y_1, \dotsc, y_m, y_1^{-1}, \dotsc, y_m^{-1}][[x]] : N \in \bZ_{>0} \}
    \]
    is said to be log-linearly bounded if for any pair $\alpha, \varepsilon > 0$, there exists $N_0 \in \bZ_{>0}$ such that
    \[
        \log A_N(\lfloor\alpha N^2\rfloor) < N^{1+\varepsilon} \quad\text{for all}\quad N > N_0.
    \]
\end{defn}

\begin{lem}\label{Lemma:Product}
    Let $(\cA, a, A)$, $(\cB, b, B)$, and $(\cC, c, C)$ be three parition triples.
    If $\cA \, \cB = \cC$, then
    \[
        A(P_1, \dotsb, P_n) \, B(P_1, \dotsb, P_n) \ge C(P_1, \dotsb, P_n)
    \]
    for all $P_1, \dotsc, P_n$.
\end{lem}
In prose, the lemma states that the cumulative absolute degeneracies of a product $\cA\cB$ are upper-bounded by the product of the cumulative absolute degeneracies of its factors $\cA, \cB$.

\begin{lem}\label{Lemma:Unflavor}
    Given a partition series
    \[ 
        \cA \in \bZ[y_1, \dotsc, y_m, y_1^{-1}, \dotsc, y_m^{-1}][[x_1, \dotsc, x_n]]
    \]
    with non-negative coefficients, the derived series 
    \[
        \cB \in \bZ[[x_1, \dotsc, x_n]]
    \]
    defined by 
    \[
        \cB(x_1, \dotsc, x_n) := \cA(x_1, \dotsc, x_n; 1, \dotsc, 1)
    \]
    has the same cumulative absolute degeneracies as $\cA$.
\end{lem}

\begin{lem}\label{Lemma:Pochhammer}
    The partition series
    \[
        \frac{1}{(x^p; x^q)_\infty} \quad\text{for all}\quad p, q \in \bZ_{>0}
    \]
    are log-linearly bounded.
\end{lem}
Whereas Lemma~\ref{Lemma:Product} and Lemma~\ref{Lemma:Unflavor} are self-evident, Lemma~\ref{Lemma:Pochhammer} is a consequence of (a slight generalization of) the Hardy-Ramanujan theorem \cite{hardy1918asymptotic}.

\paragraph{Unflavored Schur index} 

\begin{thm}\label{SBound}
    The unflavored Schur index is log-linearly bounded.
\end{thm}

\begin{proof}

The integral formula \eqref{Integral} for the unflavored Schur index was explicitly evaluated to give a closed-form expression \cite{Bourdier:2015wda},
\ie
I_\S(N, x) = \frac{1}{\vartheta_4(x^2)}\sum_{n=0}^\infty (-1)^n\left[ \binom{N+n}{N} + \binom{N+n-1}{N} \right]x^{nN+n^2},
\fe
where
\ie 
    \vartheta_4(q) = \sum_{n\in\bZ} (-1)^n q^{\frac{n^2}{2}} = (q; q)_\infty (q^{\frac12}; q)_\infty^2.
\fe 
By Lemma~\ref{Lemma:Product} and Lemma~\ref{Lemma:Pochhammer}, it suffices to prove that
\ie
    \cA_N(x) = \sum_{n=0}^\infty a(nN+n^2) x^{nN+n^2}, \quad 
    a(nN+n^2) = \binom{N+n}{N}
\fe
is log-linearly bounded.  The cumulative absolute degeneracy is 
\ie 
    A_N(M) = \sum_{n=0}^{n_\mathrm{max}} \binom{N+n}{N} = \binom{N+n_\text{max}+1}{N+1}, 
    \quad 
    n_\mathrm{max} = \lfloor \frac{\sqrt{N^2+4M}-N}{2} \rfloor,
\fe 
which for $M = \lfloor\alpha N^2\rfloor$ becomes
\ie 
    A_N(M) &= \binom{N+\lfloor\beta N\rfloor+1}{N+1}, 
    \quad 
    \beta = \frac{\sqrt{1+4\alpha}-1}{2},
\fe 
and satisfies
\ie 
    A_N(M) < 2^{N+\lfloor\beta N\rfloor+1}.
\fe 

\end{proof}

The Cardy limit and other asymptotics of the unflavored Schur index were studied in \cite{Eleftheriou:2022kkv} and no black hole entropy growth was found.

\paragraph{Flavored Schur index} 

\begin{thm}\label{FSBound}
    The flavored Schur index is log-linearly bounded.
\end{thm}

\noindent\emph{Sketch of proof.}
The key ingredient for proving Theorem~\ref{FSBound} is a closed-form expression for the flavored Schur index due to \cite{Hatsuda:2022xdv}, which we review below.\footnote{We thank Yiwen Pan for directing us to \cite{Hatsuda:2022xdv}. The fugacities in \cite{Hatsuda:2022xdv} are related to those defined here by
\ie 
    q = x^2, \quad \xi = b x^{-1}.
\fe 
}
First, define an auxiliary function
\ie 
    \theta(u, q) &:= i q^{-\frac18} 
    \vartheta_1(u, q)
    = \sum_{n\in\bZ} (-1)^n u^{n+\frac12} q^{\frac{n^2+n}{2}}
    = - u^{-\frac12} (q;q)_\infty (\tfrac{q}{u};q)_\infty (u;q)_\infty
\fe 
that satisfies
\ie\label{dtheta}
    \frac{\partial}{\partial u} \theta(u, q) |_{u=1} = (q;q)_\infty^3.
\fe 
We then write the flavored Schur index as a product
\ie 
    I_\FS(N, x; b) = \Lambda(N, x; b, u) Z(N, x; b, u),
\fe
where $u$ is an auxiliary fugacity (the dependence on $u$ cancels between the $\Lambda$ and $Z$ factors in nontrivial fashion), and
\ie 
    \Lambda(N, x; b, u)& := (-1)^N (b x^{-1})^{\frac{N^2}{2}} \frac{\theta(u, x^2)}{\theta(u b^{-N} x^N, x^2)}.
\fe 
Then $Z$ can be assembled into a grand canonical index that admits a product formula:
\ie 
    \Xi(\mu, x; b, u) &:= 1 + \sum_{N=1}^\infty Z(N, x; b, u) \mu^N
    = \prod_{p\in\bZ} \left(1-\frac{\mu b^{-p} x^p}{1-u x^{2p}}\right).
\fe

To proceed, we consider two specializations of the auxiliary fugacity, $u \to 1$ and $u = x$, which are suitable for odd and even $N$, respectively.  Suitability here means that when expanding $\Lambda$ in $x$, each series coefficient is a \emph{finite} Laurent series in $b$.  

Our strategy is to consider the grand canonical index as a series expansion in $\mu, x, b$, and \emph{majorize} it by other series with manifestly non-negative coefficients.  
We can then safely turn off the flavor fugacity $b$ in the majorizing series, which when cleverly chosen gives simple expressions in terms of $q$-Pochhammer symbols with well-established asymptotics.  Said asymptotics then rigorously bound the asymptotics of the flavored Schur index.

We now present the proof.  The reader content with the explanation above can safely skip ahead.

\begin{proof}  We treat odd and even ranks separately.

\noindent\emph{Odd rank.}  For odd $N$, consider the $u \to 1$ limit. Using \eqref{dtheta}, we have
\ie 
    \lim_{u\to1} \theta(u, x^2) \, \Xi(\mu, x; b) = \mu \times (x^2;x^2)_\infty^3 \prod_{p\neq0} \left(1-\frac{\mu b^{-p} x^p}{1-x^{2p}}\right),
\fe 
and hence we can write
\ie 
    I_\FS(N, x; b) = \tilde\Lambda(N, x; b) \tilde Z(N, x; b),
\fe 
where
\ie 
    \tilde\Lambda(N, x; b) &:= 
    x^{\frac{N^2-1}{4}} \times
    (-1)^N (bx^{-1})^{\frac{N^2}{2}} \frac{(x^2;x^2)_\infty^3}{\theta(b^{-N} x^N, x^2)},
    \\
    \tilde Z(N, x; b) &:= x^{-\frac{N^2-1}{4}} \times \frac{\mu}{(x^2;x^2)_\infty^3} \lim_{u\to1} \theta(u, x^2) Z(N, x; b, u).
\fe 
Note that we have inserted factors of $x^{\frac{N^2-1}{4}}$ such that both $\tilde\Lambda$ and $\tilde Z$ are partition series, i.e.\ their expansions in $x$ start from $x^0$.
A grand canonical index is given by
\ie\label{OddXiForm}
    \tilde\Xi(\mu, x; b) &:= \sum_{N=1}^\infty x^{\frac{N^2-1}{4}} \tilde Z(N, x; b) \mu^N
    \\
    &= \mu \prod_{p\neq0} \left(1-\frac{\mu b^{-p} x^p}{1-x^{2p}}\right) 
    = \mu \prod_{p>0} \left(1-\frac{\mu b^{-p} x^{p}}{1-x^{2p}}\right)
    \left(1+\frac{\mu b^{p} x^{p}}{1-x^{2p}}\right).
\fe 
By Lemma~\ref{Lemma:Product}, we would have proven Theorem~\ref{FSBound} if we could prove that $\tilde\Lambda$ and $\tilde Z$ are each log-linearly bounded.

First, we deal with the prefactor $\tilde\Lambda$, which upon some manipulations becomes
\ie 
    \tilde\Lambda(N, x; b) &= 
    (-1)^{\frac{N-1}{2}} 
    \frac{(x^2;x^2)_\infty^2}{(b^{-N} x^N;x^2)_\infty (b^N x; x^2)_\infty} 
    \prod_{n=0}^{\frac{N-3}{2}} \frac{1}{1-b^{-N} x^{2n+1}}.
\fe 
It is obvious that
\ie 
    \tilde\Lambda(N, x; b) \preceq \frac{
    1
    }{(b^{-N} x^N;x^2)_\infty (b^N x; x^2)_\infty (x^2;x^2)_\infty^2 (b^{-N} x;x^2)_\infty},
\fe 
where the right side has non-negative series coefficients. 
A combination of Lemma~\ref{Lemma:Product},~\ref{Lemma:Unflavor},~\ref{Lemma:Pochhammer} then shows that $\tilde\Lambda$ is log-linearly bounded.

Next, we deal with $\tilde Z$.  It is obvious from \eqref{OddXiForm} that
\ie  
    \tilde\Xi(\mu, x; b) \preceq \tilde\Omega(\mu, x; b) := \mu \prod_{p>0} \left(1+\frac{\mu b^{-p} x^{p}}{1-x^{2p}}\right) \left(1+\frac{\mu b^{p} x^{p}}{1-x^{2p}}\right),
\fe
where $\tilde\Omega$ has manifestly non-negative coefficients.  By Lemma~\ref{Lemma:Unflavor}, to show that
\ie
    \tilde Z(N, x; b) = x^{-\frac{N^2-1}{4}} \tilde\Xi(\mu, x; b)|_{\mu^N}
\fe
is log-linearly bounded, it suffices to show that $x^{-\frac{N^2-1}{4}} \tilde\Omega(\mu,x;1)|_{\mu^N}$ is log-linearly bounded, where
\ie 
    \tilde\Omega(\mu,x;1) = \mu \, \tilde\omega(\mu,x)^2, \quad \tilde\omega(\mu,x) &:= \prod_{p>0} \left(1+\frac{\mu x^{p}}{1-x^{2p}}\right).
\fe 
This can be achieved by another sequence of majorizations,
\ie 
    \tilde\omega(\mu;x)
    \preceq
    \tilde\nu(\mu;x) &:=
    \prod_{p>0} \frac{1 + \mu x^{p}}{1-x^{2p}} = \frac{(-\mu x,x)_\infty}{(x^2;x^2)_\infty}
    \\
    &
    = \frac{1}{(x^2;x^2)_\infty} \sum_{N=0}^\infty \frac{x^{\frac{N(N+1)}{2}}}{(x; x)_N} \mu^N
    \preceq
    \frac{1}{(x^2;x^2)_\infty} \sum_{N=0}^\infty \frac{x^{\frac{N(N+1)}{2}}}{(x; x)_\infty} \mu^N,
\fe 
resulting in
\ie 
    x^{-\frac{N^2-1}{4}} \tilde\Omega(\mu,x;1)|_{\mu^N} 
    &\preceq 
    \frac{x^{-\frac{N^2-1}{4}}}{(x^2;x^2)_\infty^2 (x;x)_\infty^2} 
    \sum_{p+q=N-1} x^{\frac{p(p+1)}{2}+\frac{q(q+1)}{2}} 
    \preceq 
    \frac{N}{(x^2;x^2)_\infty^2 (x;x)_\infty^2},
\fe 
and finally invoking Lemma~\ref{Lemma:Pochhammer}.

\noindent\emph{Even rank.}  For even $N$, set $u = x$, such that 
\ie 
    \tilde\Lambda(N, x; b) &:= x^{\frac{N(N-2)}{4}} \times \Lambda(N, x; b, x)
    \\
    &= x^{\frac{N(N-2)}{4}} \times (-1)^N b^{\frac{N^2}{2}} x^{-\frac{N^2}{2}} \frac{\theta(x, x^2)}{\theta(b^{-N} x^{N+1}, x^2)}
    \\
    &= (-1)^{\frac{N}{2}} \frac{b^{\frac{N}{2}} (x; x^2)_\infty^2}{(b^N x; x^2)_\infty (b^{-N} x^{N+1}; x^2)_\infty} \prod_{n=0}^{\frac{N}{2}-1} \frac{1}{(1-b^{-N} x^{2n+1})},
    \\
    \tilde Z(N, x; b) &:= x^{-\frac{N(N-2)}{4}} \times Z(N, x; b, x),
    \\
    \tilde\Xi(\mu, x; b) &:= \sum_{N=0}^\infty x^{\frac{N(N-2)}{4}} \tilde Z(N, x; b) \mu^N
    \\
    &= \prod_{p\in\bZ} \left(1-\frac{\mu b^{-p} x^p}{1-x^{2p+1}}\right)
    = 
    \prod_{p \ge 0} \left(1-\frac{\mu b^{-p} x^{p}}{1-x^{2p+1}}\right) \left(1+\frac{\mu b^{p+1} x^{p}}{1-x^{2p+1}}\right).
\fe
The proof then proceeds in the same way as for odd rank.
\end{proof}

\subsection{Explicit evaluation}
\label{sec:computation_index}

In \cite{Murthy:2020rbd, Agarwal:2020zwm}, the unrefined SCI capturing $\frac{1}{16}$-BPS states in $\cN=4$ SYM for gauge group $\U(N)$ was explicitly evaluated.  In particular, \cite{Murthy:2020rbd} used symmetric group theory and a formula derived from the integral formula, which is well-suited for evaluating the various other indices introduced in Section~\ref{Sec:Indices}.\footnote{The required data for $N=1,2,\dotsc,10$ were extracted from the character tables of $\S_n$ up to $n=50$ using GAP \cite{GAP4}, and are publicly available in Mathematica \texttt{.m} format in the \texttt{sn/} folder of \href{https://github.com/yinhslin/bps-counting/}{\texttt{https://github.com/yinhslin/bps-counting/}}. Also provided is a Mathematica notebook \texttt{index.nb} illustrating how the symmetric group data is used in computing the microcanonical indices.}
We evaluate the absolute degeneracies, defined in the notation of Definition~\ref{Def:PartitionSeries} as
\ie 
    |d(2h)| := \sum_{q_1, \dotsc, q_m} |a(2h; q_1, \dotsc, q_m)|,
\fe 
for the unflavored and flavored Schur indices, the flavored MacDonald index, as well as a specialized version of the SCI,
\ie\label{SpecializedSCI}
    I_\SCI(x, s, b=1, p=xs^{-1}).
\fe 
Note that the unrefined SCI capturing $\frac{1}{16}$-BPS states and studied in \cite{Murthy:2020rbd, Agarwal:2020zwm} was
\ie\label{UnrefinedSCI}
    I_\SCI(x=t^\frac{3}{2},s=t^\frac{1}{2},b=1,p=t^3).
\fe 
The fact that \eqref{SpecializedSCI} when specialized to the Schur sector with $q_3 + J_2 = 0$ has the same exponent $2J_1 + q_1 + q_2$ of $x$ makes its comparison with the Schur and MacDonald indices more natural than \eqref{UnrefinedSCI}.

In Tables~\ref{Tab:U},~\ref{Tab:SUEven},~\ref{Tab:SUOdd} in Appendix~\ref{App:Tables}, we compare the absolute degeneracies of the unflavored and flavored Schur indices.  For $\U(N)$ gauge groups they are extremely similar at low levels.  In fact, up to $2h=30$, only $\U(3)$ exhibits a difference starting at $2h=18$.  We expect that such differences should appear at higher levels for $N>4$, but this is out of the scope of our explicit evaluation.  These explicit degeneracies at low levels complement the previous asymptotic analyses, confirming that the entropy growth exhibited by the unflavored and flavored Schur indices are very similar.

Figures~\ref{Fig:U} and~\ref{Fig:SU} present the comparison of the flavored Schur index, the flavored MacDonald index, and the specialized SCI.
As was also observed in \cite{Murthy:2020rbd, Agarwal:2020zwm} for the unrefined SCI \eqref{UnrefinedSCI}, for fixed charges the degeneracies exhibit little growth with increasing $N$, and the main source of black hole entropy comes from the log-linear growth in the charges, e.g.\ $h$, for fixed $N$, resulting in log-quadratic growth in $N$ as we take charges to scale as $N^2$.
It can be seen that the growth of the flavored MacDonald index is on par with that of the flavored Schur, especially for $\SU(N)$ gauge group, whereas the specialized SCI exhibits exponentially faster growth.
In particular, the specialized SCI shows approximately log-linear growth in $h$ for fixed $N$, whereas both the flavored Schur and MacDonald indices exhibit log-sublinear (supposedly log-square-root) growth.

\newpage

\begin{figure}[H]
	\centering
	\includegraphics[width=0.35\textwidth]{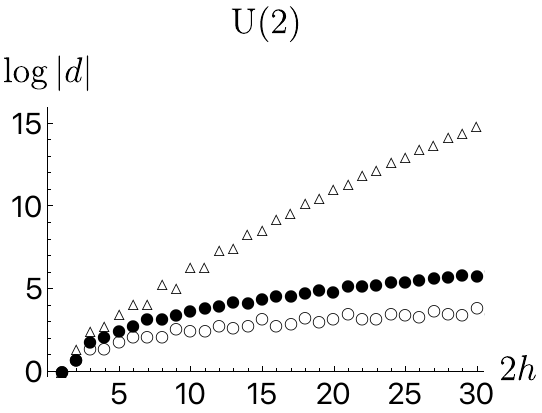} \qquad
	\includegraphics[width=0.35\textwidth]{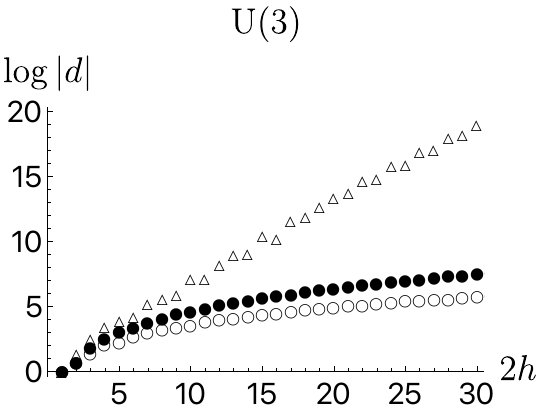}
	\\
	~\vspace{-0.1in}
	\\
	\includegraphics[width=0.35\textwidth]{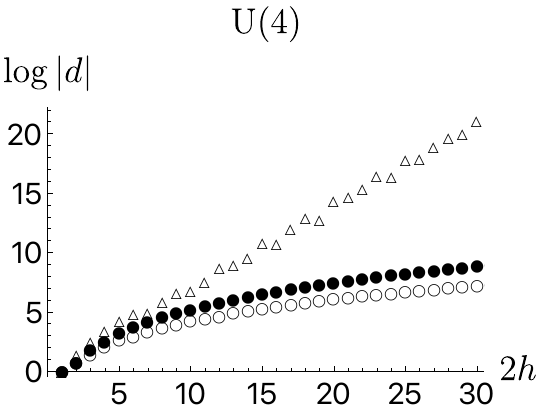} \qquad
	\includegraphics[width=0.35\textwidth]{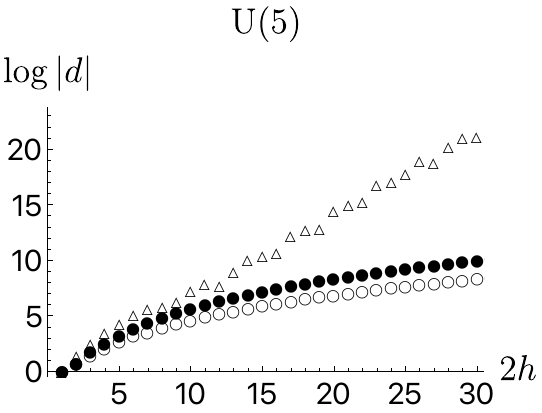}
	\\
	~\vspace{-0.1in}
	\\
	\includegraphics[width=0.35\textwidth]{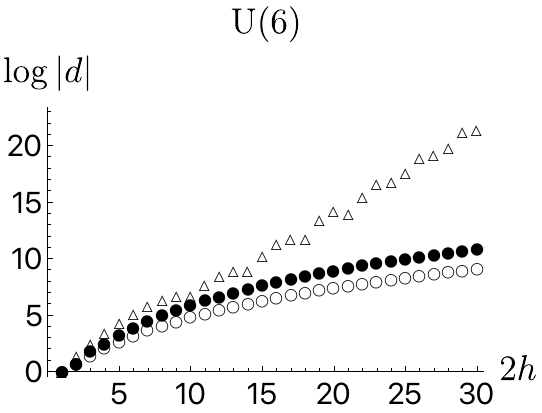} \qquad
	\includegraphics[width=0.35\textwidth]{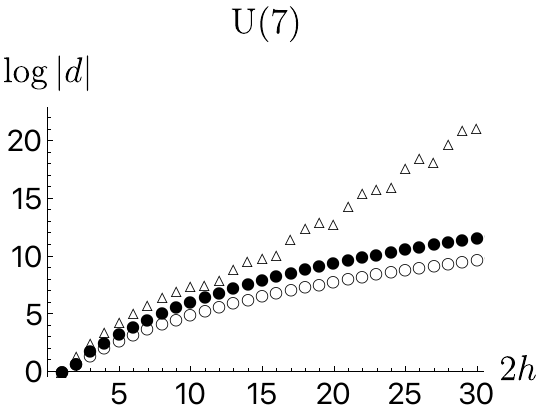}
	\\
	~\vspace{-0.1in}
	\\
	\includegraphics[width=0.35\textwidth]{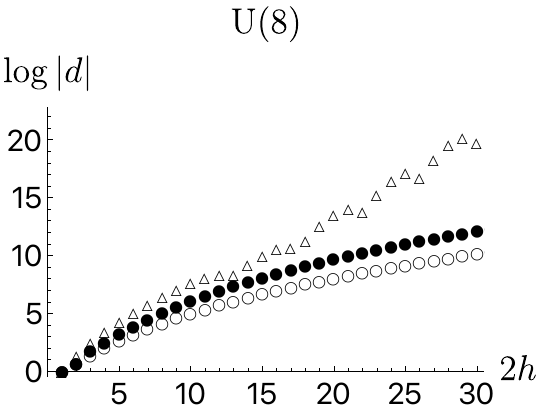} \qquad
	\includegraphics[width=0.35\textwidth]{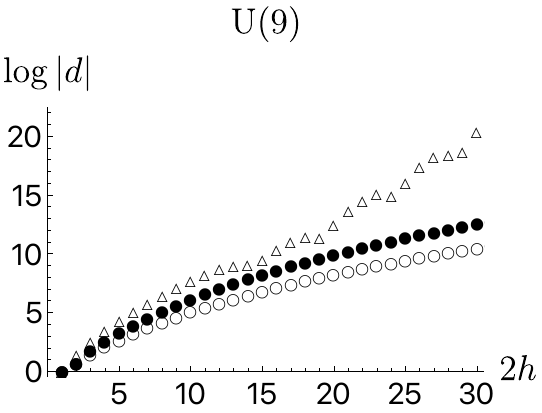}
	\\
    ~\vspace{-0.1in}
	\\
	\caption{Absolute degeneracies of the flavored Schur index (empty circle), the flavored MacDonald index (solid circle), and a specialized SCI (empty triangle), up to $2h = 30$ for $\U(N)$ gauge groups with $N = 2, \dotsc, 9$.}
	\label{Fig:U}
\end{figure}

\newpage

\begin{figure}[H]
	\centering
	\includegraphics[width=0.35\textwidth]{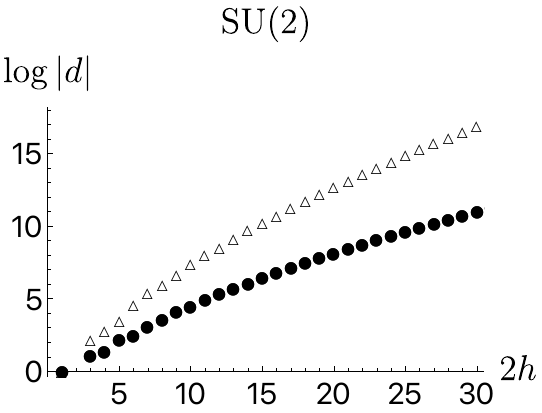} \qquad
	\includegraphics[width=0.35\textwidth]{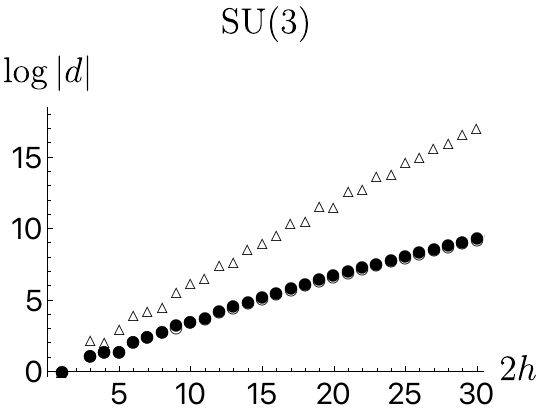}
	\\
	~\vspace{-0.1in}
	\\
	\includegraphics[width=0.35\textwidth]{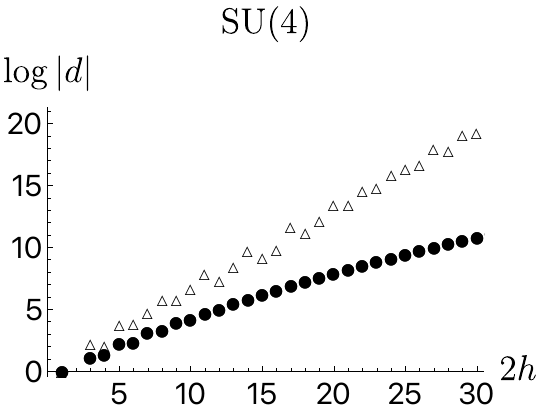} \qquad
	\includegraphics[width=0.35\textwidth]{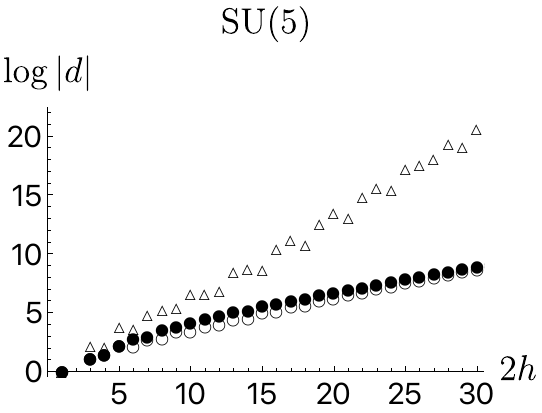}
	\\
	~\vspace{-0.1in}
	\\
	\includegraphics[width=0.35\textwidth]{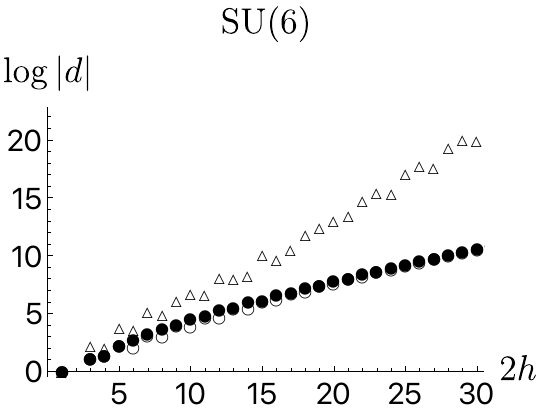} \qquad
	\includegraphics[width=0.35\textwidth]{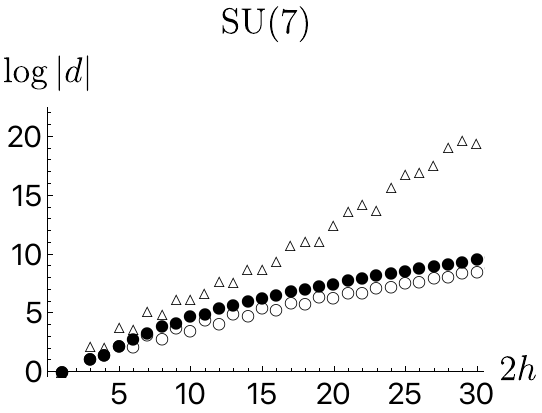}
	\\
	~\vspace{-0.1in}
	\\
	\includegraphics[width=0.35\textwidth]{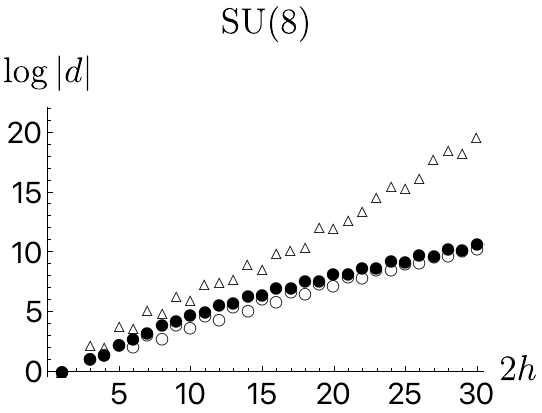} \qquad
	\includegraphics[width=0.35\textwidth]{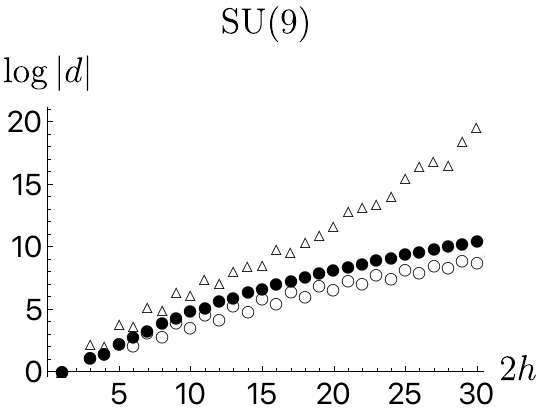}
	\\
	~\vspace{-0.1in}
	\\
	\caption{Absolute degeneracies of the flavored Schur index (empty circle), the flavored MacDonald index (solid circle), and a specialized SCI (empty triangle), up to $2h = 30$ for $\SU(N)$ gauge groups with $N = 2, \dotsc, 9$.}
	\label{Fig:SU}
\end{figure}

\section{Operators and algebras}
\label{Sec:OpAlg}

Having established strong evidence against black hole entropy growth in the Schur sector by studying indices in the previous section, we now turn to analyzing the algebraic structure of the chiral algebra, to assess the full validity of Conjecture~\ref{Conj} and the non-existence of non-gravitons.

\subsection{Correlation functions and the operator product expansion}
\label{Sec:Closure}

We begin by studying the correlation functions and the OPEs of the multi-graviton operators and focus on the conjectured \cite{Beem:2013sza} strong generators of the chiral algebra of the $\cN=4$ SYM.
In light of the discussion at the end of Section~\ref{sec:holotwist}, 
composite operators must be defined by
the regularized product $(\cdot,\cdot)$, which is the leading regular term in the OPE. The chiral algebra generated is conjectured to be an $\mathcal{N} = 4$ generalization of the $\mathcal{W}$-algebra. While it is widely believed to exist, little is known about this $\mathcal{N} = 4$ $\mathcal{W}$-algebra for gauge rank $N>3$.\footnote{See \cite{Bonetti:2018fqz,Arakawa:2023cki} for an interesting alternative free field construction.}
Here, we intend to make some progress on understanding the OPE of the strong generators at generic $N$. For the reasons that will become clear below, we find it more convenient to work with gauge group $\U(N)$.

The symplectic bosons $q_i^A$ form a doublet of the $\SL(2,\mathbb{R})_f$ global symmetry, and have the OPE
\begin{equation}
	q_i^A(z_1) q_j^B(z_2) \sim -\epsilon_{ij} \delta^{AB} \frac{1}{z_{12}}, \quad i,j = 1,2,
\end{equation}
where the capital letters $A,B = 1, \dots, \mathrm{dim}\ \mathfrak{g}$ are gauge indices, and $i,j = 1,2$ are $\SL(2,\mathbb{R})_f$ indices. Following \cite{Bonetti:2018fqz}, let us introduce an index-free notation
\begin{equation}
	Z(u ; z) := q_1(z)+u q_2(z) \label{eq:Z=q+uq}
\end{equation}
using an auxiliary parameter $u$, such that the OPE becomes
\begin{equation}
	Z^A(u_1;z_1) Z^B(u_2;z_2) \sim \frac{u_{12}}{z_{12}}\delta^{AB},
\end{equation}
where $u_{12} := u_1-u_2$ and $z_{12} := z_1-z_2$. More generally, suppose $O_{0}(z), O_{1}(z), \dots, O_{2j}(z)$ furnish a spin $j$ representation, we construct
\begin{equation}
    O(u;z) := \sum_{i=0}^{2j} u^{i} O_i(z),
\end{equation}
which transforms as
\begin{equation}
{O}^{\prime}\left(u^{\prime};z\right)=\left(\frac{\partial u^{\prime}}{\partial u}\right)^j {O}(u;z),
\end{equation}
and $u^\prime$ is the M\"obius transformation of $u$,
\ie
u^{\prime}=\frac{\hat{a} u+\hat{b}}{\hat{c} u+\hat{d}}, \quad\left(\begin{array}{cc}
	\hat{a} & \hat{b} \\
	\hat{c} & \hat{d}
\end{array}\right) \in \mathrm{SL}(2).
\fe

The Lie algebra cohomology \eqref{eq:liealgebracoh} can be explicitly evaluated in the large $N$ limit with the help of the Loday-Quillen-Tsygan theorem \cite{Costello:2018zrm}. This nice result can be succinctly summarized as follows. Define a differential on $\mathbb{C}^{1|2}$ by $\mathrm{d} = \mathrm{d}z\partial_z + \mathrm{d}\theta^i \partial_{\theta^i}$, acting on the superfield
\ie
\label{C}
C[z,x_i,\theta^1,\theta^2]  := c(z,x_i) + \theta^i q_i(z,x_i) + \theta^1\theta^2 b(z,x_i) \in \Omega_{\mathbb{C}}^{0,0} \otimes \Omega_{\mathbb{R}^2}^{0} \otimes \mathfrak{g}[\theta_1,\theta_2].
\fe
defined earlier in \eqref{eq:Csuperfield}. Single-graviton operators are the $z$-derivatives of expansion coefficients of (with $\epsilon_{21} = 1$) 
\begin{equation}
\begin{aligned}
    \mathrm{Tr} (\mathrm{d} C)^n & =  \mathrm{d}\theta^{(i_1}\mathrm{d}\theta^{i_2}\dots \mathrm{d}\theta^{i_n)} \left(A_n\right)_{i_1\dots i_n}\\
    &+ n(\epsilon_{ij}\theta^{i}\mathrm{d}\theta^{j}  )\mathrm{d}\theta^{(i_1}\mathrm{d}\theta^{i_2}\dots \mathrm{d}\theta^{i_{n-1)}} \left( B_{n} \right)_{i_1 \dots i_{n-1}}\\
    &+ n \mathrm{d}z \mathrm{d}\theta^{(i_1}\mathrm{d}\theta^{i_2}\dots \mathrm{d}\theta^{i_{n-1})} \left( C_{n} \right)_{i_1 \dots i_{n-1}} \\
    &+ n\mathrm{d}z  \epsilon_{ij}\theta^i \mathrm{d}\theta^{(j}   \mathrm{d}\theta^{i_1}\dots \mathrm{d}\theta^{i_{n-2})} \left(D_{n} \right)_{i_1\dots i_{n-2}} \\
    & +n \mathrm{d}z \theta^{(i_1} d\theta^{i_2}  \mathrm{d}\theta^{i_3}\dots \mathrm{d}\theta^{i_{n})} \left( \partial A_n \right)_{i_1\dots i_n} \\
    & + n( \mathrm{d}z \theta^1\theta^2) \mathrm{d}\theta^{i_1}\dots \mathrm{d}\theta^{i_{n-1}} \left( \partial B_{n} \right)_{i_1\dots  i_{n-1}},
\end{aligned}
\end{equation}
which includes
\begin{equation}
    A_n=\Tr q_{\left(i_1\right.} q_{i_2} \cdots q_{\left.i_n\right)} \label{eq:Antower}
\end{equation}
and the superconformal partners
\begin{equation}
	\begin{aligned}
		& B_{n}=\operatorname{Tr} b q_{\left(i_1\right.} q_{i_2} \cdots q_{\left.i_{n-1}\right)}, \\
		& C_{n}=\operatorname{Tr} \partial c q_{\left(i_1\right.} q_{i_2} \cdots q_{\left.i_{n-1}\right)}, \\
		& D_{n}= \epsilon^{i j} \operatorname{Tr} \partial q_i q_{(j} q_{i_1}\cdots q_{\left.i_{n-2}\right)}+\Tr b\partial c q_{(i_1} \dots q_{i_{n-2})}+\Tr b q_{(i_1} \partial c q_{i_2} \dots q_{i_{n-2})}+ \dots \\
  & \phantom{asdfas} + \Tr b q_{(i_1}  q_{i_2} \dots q_{i_{n-2})} \partial c.
	\end{aligned}\label{eq:BCDtower}
\end{equation}
It is easy to see that this recovers \eqref{eqn:Schur_grav_op} upon identifying $C$ with $\Psi$.

From \eqref{eq:Antower}, the A tower operators $\left(A_n\right)_{i_1i_2\dots i_n}$ have scaling dimensions $h = \frac{n}{2}$, and form a spin $j = \frac{n}{2}$ multiplet of $\SL(2,\mathbb{R})_f$, so we have
\begin{equation}
	A_n(u ; z)= \frac{1}{N^{n/2}}\operatorname{Tr} Z(u ; z)^n\label{eq:Anuz},
\end{equation}
where a normalization factor is introduced for later convenience. 
Similarly, the B and C tower operators $B_n$ and $C_n$ have $h = \frac{n}{2} +\frac12$ and $j = \frac{n}{2} - \frac12$, and give rise to
\begin{equation}
    \begin{aligned}
    B_n(u;z) &= \frac{1}{N^{n/2}} \Tr b Z(u;z)^{n-1},\\
    C_n(u;z) &= \frac{1}{N^{n/2}} \Tr \partial c Z(u;z)^{n-1}.
     \end{aligned}\label{eq:BCnuz}
\end{equation}
Finally the D tower operators $D_n$ have $h = \frac{n}{2} + 1$ and $j= \frac{n}{2} - 1$, and can be found in the expansion of 
\ie
    D_n(u;z)
    & = \frac{1}{N^{n/2}} \epsilon^{ij} \sum_{k = 0}^{n-2}\Tr \partial q_i Z(u;z)^k q_j Z(u;z)^{n-2-k} 
    \\
    & \qquad + \frac{1}{N^{n/2}} \sum_{k = 0}^{n-2}\Tr b Z(u;z)^{k} \partial c Z(u;z)^{n-2-k}. \label{eq:Dnuz}
\fe

\paragraph{Correlation functions in the planar limit} We first extend the calculation of \cite{Costello:2018zrm} to superfields. Generalizing \eqref{eq:Z=q+uq} to the whole superfield $C$ of \eqref{C}, we define
\begin{equation}
     \mathrm{d}C  = q_1 + uq_2 + (\theta^2-\theta^1u)b +\mathrm{d}z\partial c + \mathrm{d}z \theta^i \partial q_i + \mathrm{d}z \theta^1\theta^2\partial b,
\end{equation}
where we replaced $d\theta^1$ and $d\theta^2$ by $1$ and $u$, respectively. 
The OPE of the free fields can be organized into the form
\begin{equation}
	\begin{aligned}
	\mathrm{d}C(u;z,\theta) \mathrm{d}C({u}^\prime;{z}^\prime, {\theta}^\prime) &\sim  \frac{u-{u^\prime}}{z-{z^\prime}}
	 - \left(\mathrm{d}z-\mathrm{d}{z^\prime}\right) \frac{(\theta^2-{\theta^\prime}^2)}{(z-{z^\prime})^2}+ \left({u^\prime}\mathrm{d}z-u\mathrm{d}{z^\prime}\right) \frac{(\theta^1-{\theta^\prime}^1)}{(z-{z^\prime})^2}\\
	& +2 \mathrm{d}z\mathrm{d} {z^\prime} \left( \frac{(\theta^1-{\theta^\prime}^1)(\theta^2-{\theta^\prime}^2)}{(z-{z^\prime})^3} \right). \end{aligned}
    \label{eq:dCdCOPE}
\end{equation}
To reduce clutter, we use the shorthand notation $\underline{u}:= (u;z,\theta)$ and write the right-hand side of \eqref{eq:dCdCOPE} as $\mathcal{P}(\underline{u};\underline{u}')$. The correlation functions of $\Tr (dC)^n$ can be simply obtained by Wick contractions, and become particularly simple in the planar limit.

The two-point functions $\langle \Tr dC(\underline{u})^n \Tr dC(\underline{u}')^n\rangle$ are computed by planar pairwise Wick contractions, each of which is a closed loop and contributes a factor of $N$ by the contraction of delta functions. There are $n$ different ways to do so, so we have
\begin{equation}
     \langle \Tr (dC(\underline{u}))^n \Tr (dC(\underline{u}'))^n\rangle =N^n n \mathcal{P}(\underline{u},\underline{u}')^n.
\end{equation}
Notice the scaling power of $N^n$. 
Similarly, three-point functions
\begin{equation}
    \langle \Tr (dC(\underline{u}))^n \Tr (dC(\underline{u}'))^{n'} \Tr (dC(\underline{u}''))^{n''}\rangle
\end{equation}
involve $\frac{n+n'-n''}{2}$ Wick contractions between the first two operators, $\frac{n+n''-n'}{2}$ Wick contractions between the first, and the third operators and $\frac{n'+n''-n}{2}$ Wick contractions between the last two operators. Hence,
\begin{equation}
\begin{aligned}
     \langle \Tr & (dC(\underline{u}))^n \Tr (dC(\underline{u}'))^{n'} \Tr (dC(\underline{u}'') )^{n''}\rangle \\
     = & \frac{nn'n''}{N} N^{\frac{n+n'+n''}{2}}\mathcal{P}(\underline{u},\underline{u}')^{\frac{n+n'-n''}{2}}\mathcal{P}(\underline{u}',\underline{u}'')^{\frac{n'+n''-n}{2}}\mathcal{P}(\underline{u},\underline{u}'')^{\frac{n+n''-n'}{2}}.
     \end{aligned}
\end{equation}
Expanding in $dz$, $\theta^i$, we can read off the two- and three-point functions of the ABCD towers. For example, we find the two-point functions
\begin{equation}
\left\langle A_{n}(u ; z) A_{n}\left(u^{\prime} ; z^{\prime}\right)\right\rangle=n \frac{\left(u-u^{\prime}\right)^n}{\left(z-z^{\prime}\right)^n}, \label{eq:AnAn2pt}
\end{equation}
and
\begin{equation}
    \left\langle D_{n}(u ; z) D_{n}\left(u^{\prime} ; z^{\prime}\right)\right\rangle =  \frac{\left(u-u^{\prime}\right)^{n-2}}{\left(z-z^{\prime}\right)^{n-2}} \frac{n(n-1)}{(z-z')^4},
\end{equation}
as well as the three-point functions
\ie
& \left\langle A_{n}(u; z) A_{n^{\prime}}\left(u^{\prime} ; z^{\prime}\right) A_{n^{\prime \prime}}\left(u^{\prime \prime} ; z^{\prime \prime}\right)\right\rangle
\\
&= \frac{n n^{\prime} n^{\prime \prime}}{N} 
 \frac{\left(u-u^{\prime}\right)^{\frac{n+n^{\prime}-n^{\prime \prime}}{2}}\left(u^{\prime}-u^{\prime \prime}\right)^{\frac{n^{\prime}+n^{\prime \prime}-n}{2}}\left(u^{\prime \prime}- u\right)^{\frac{n^{\prime \prime}+n-n^{\prime}}{2}}}{\left(z-z^{\prime}\right)^{\frac{n+n^{\prime}-n^{\prime \prime}}{2}}\left(z^{\prime}-z^{\prime \prime}\right)^{\frac{n^{\prime}+n^{\prime \prime}-n}{2}}\left(z^{\prime \prime}-z\right)^{\frac{n^{\prime \prime}+n-n^{\prime}}{2}}}, \label{eq:3ptAAA}
\fe
and 
\ie
& \left\langle A_{n}(u; z) A_{n^{\prime}}\left(u^{\prime} ; z^{\prime}\right) D_{n^{\prime \prime}}\left(u^{\prime \prime} ; z^{\prime \prime}\right)\right\rangle 
\\
&=
\frac{n n^{\prime} }{N} 
\frac{ \left(u-u^{\prime}\right)^{\frac{n+n^{\prime}-n^{\prime \prime}}{2}+1} }{\left(z-z^{\prime}\right)^{\frac{n+n^{\prime}-n^{\prime \prime}}{2}-1}}
\frac{\left(u^{\prime}-u^{\prime \prime}\right)^{\frac{n^{\prime}+n^{\prime \prime}-n}{2}-1}}{\left(z^{\prime}-z^{\prime \prime}\right)^{\frac{n^{\prime}+n^{\prime \prime}-n}{2}+1}}
\frac{ \left(u^{\prime \prime}- u\right)^{\frac{n^{\prime \prime}+n-n^{\prime}}{2}-1} }{\left(z^{\prime \prime}-z\right)^{\frac{n^{\prime \prime}+n-n^{\prime}}{2}+1}}. \label{eq:3ptAAD}
\fe
In particular, the A tower correlation functions agree with the computation \cite{Costello:2018zrm}.

\paragraph{OPE in the planar limit} The OPE between operators of definite spins and scaling weight are heavily constrained by $\SL(2,\mathbb{R})_f$ covariance and dimensional analysis, and in general takes the following form
\begin{equation}
\mathcal{O}_1\left(u_1;z_1\right) \mathcal{O}_2\left(u_2;z_2\right)=\sum_{\mathcal{O}} \lambda_{\mathcal{O}_1 \mathcal{O}_2}^{\mathcal{O}} \frac{u_{12}^{j_1+j_2-j} }{z_{12}^{h_1+h_2-h}} \mathcal{D}_{h_1, h_2 ; h}\left(z_{12}, \partial_{z_2}\right) \widehat{\mathcal{D}}_{j_1, j_2 ; j}\left(y_{12}, \partial_{y_2}\right) \mathcal{O}_j\left(z_2\right),
\end{equation}
where $\mathcal{O}_i$ has spin $j_i$ and dimension $h_i$, whereas $\mathcal{D}$ and $\widehat{\mathcal{D}}$ are the so-called weight-lifting operators \cite{Bonetti:2018fqz}
\ie
\mathcal{D}_{h_1, h_2 ; h}\left(z_{12}, \partial_{z_2}\right) &= \sum_{k=0}^{\infty} \frac{\left(h+h_1-h_2\right)_k}{k !(2 h)_k} z_{12}^k \partial_{z_2}^k,
\\
\widehat{\mathcal{D}}_{j_1, j_2 ; j}\left(y_{12}, \partial_{y_2}\right) &= \sum_{k=0}^{2j} \frac{\left(-j-j_1+j_2\right)_k}{k !(-2 j)_k} y_{12}^k \partial_{y_2}^k.
\fe

In the planar limit, the OPE of single trace operators only produces single trace operators, i.e.\ the ABCD towers and their derivatives, as the multi-trace operators are suppressed by powers of $1/N$. Let's focus on the OPE between the A tower operators. Since the B and C tower operators are fermionic, only the A and D tower operators can contribute. Their OPE coefficients can be obtained from three-point functions. Following the reduced notation where we omit $\mathcal{D}$, $\widehat{\mathcal{D}}$, powers of $z_{12}$ and $y_{12}$, and keep only the OPE coefficients, we find
\begin{equation}
    A_n  A_{n'} \sim \sum_{n''}
    \frac{n n^{\prime} }{N} A_{n''} +\sum_{n'''} 
    \frac{nn'}{N}\frac{1}{n'''(n'''-1)} D_{n'''},
\end{equation}
where the sum ranges over
\ie
   n'' &= |n-n'|, |n-n'|+2, \dots, |n+n'|-2,
   \\
   n''' &= |n-n'|+2,|n-n'|+4 \dots, |n+n'|-2.
\fe
In particular, we see that
\begin{equation}
	\begin{aligned}
		A_2 A_n &\sim \frac{2n}{N} \Big(  A_{n-2} +  A_n \Big), \\
		A_3 A_n &\sim \frac{3n}{N}\Big(  A_{n-3} +  A_{n-1} + A_{n+1} + \frac{1}{(n-1)(n-2)} D_{n-1}  \Big), \quad n> 3.
	\end{aligned}
\end{equation}
A remark is in order. Unlike $N=2$, for $N\geq 3$, the A tower operators do not close among themselves and one finds D tower operators in the OPE of A tower operators. On the other hand, the canonical representatives \eqref{eq:Dnuz} of the D tower operators contain $b \partial c$. Therefore, the closure of the OPE suggests one can always find representatives of the D tower operators (in the cohomology) that only contain $q_i$ and derivatives. We will explore this interesting prediction in the future. 

\paragraph{OPE at finite $N$} It is instructive to compare the above with the OPE at finite $N$. According to the conjecture of \cite{Beem:2013sza}, the strong generators are given by a finite number of ABCD towers with $n\leq N$, while operators with $n>N$ are generated in the OPE due to the trace relations. It is however much more difficult to calculate correlation functions analytically at generic $N$ since we need to consider all non-planar diagrams. In addition, the OPE of strong generators receives contributions from multi-trace operators built out of normal-ordered single-trace operators. We will leave this study to the future. Instead, we will perform explicit calculations at small $n$ and small ranks ($N=2$, 3, 4, 5) using the $\texttt{OPEdefs}$ package\cite{Thielemans:1991uw}. 

Let us start by discussing the two-point functions. The $\langle A_1A_1\rangle$, $\langle A_2 A_2 \rangle$ correlators are still given by \eqref{eq:AnAn2pt}, but all other correlators receive $1/N$ corrections. For example, we find
\begin{equation}
\left\langle A_{3}(u ; z) A_{3}\left(u^{\prime} ; z^{\prime}\right)\right\rangle=3 \left(1+ \frac{1}{N^2}\right)  \frac{\left(u-u^{\prime}\right)^n}{\left(z-z^{\prime}\right)^n}.
\end{equation}
The $A_1 A_n$ OPE does not receive multi-trace corrections, and we have
\begin{equation}
    A_1A_n \sim  \frac{n}{N} A_{n-1}.
\end{equation}
Similarly, the $A_2A_2$ and $A_2A_3$ OPEs have the same $N$ dependence as the large $N$ limit
\ie
    A_2 A_2 \sim \frac{4}{N}\left(A_0+A_2 \right),
    \quad
    A_2A_3 \sim \frac{6}{N} \left(A_1+A_3\right).
\fe
However, $A_2A_4$ is corrected by multi-trace operators
\begin{equation}
    A_2 A_4 \sim \frac{8}{N} \left( A_2 + \frac{1}{2N} (A_1A_1)_{j=1} +A_4\right).
\end{equation}
As we increase $n$ and $n'$, the OPE quickly gets more complicated. For example,
\begin{equation}
\begin{aligned}
    A_3A_3 \sim \frac{9}{N} &\left[\left(1+\frac{1}{N^2} \right)A_0 + \left( A_2 +  \frac{1}{N} (A_1A_1)_{j=1}  \right) \right.\\
    & +\left( A_4 + \frac{1}{2N}(A_2A_2)_{j=0}  +\frac{1}{4N}(A_1A_1A_1A_1)_{j=0}  \right)
   +\left. \left(\frac12- \frac{1}{N^2}\right) D_2
    \right].
    \end{aligned}
    \label{eq:A3A3OPE}
\end{equation}
Note that the multi-trace operators arising as regularized products of $A_n$ do not have definite spins under $\SL(2,\mathbb{R})_f$. Therefore we use $({O}_1\dots O_m)_j$ to denote the projection of the multi-trace operator $({O}_1\dots O_m)$ to its irreduble multiplet of spin $j$. Given two operators in the index-free notation,
\begin{equation}
    {O}_1(u)=\sum_{k_1=0}^{2 j_1} {O}_{1, k_1} u^{k_1}, \quad {O}_2(u)=\sum_{k_2=0}^{2 j_2} {O}_{1, k_2} u^{k_2}.
\end{equation}
The projection of their normal ordered product to spin $j$ is given by \cite{Bonetti:2018fqz}\footnote{Notice the typos in \cite{Bonetti:2018fqz}. We thank Federico Bonetti for useful discussions on this.}
\begin{equation}
    \left({O}_1, {O}_2\right)_j\left(u\right)=\sum_{k_1=0}^{2 j_1} \sum_{k_2=0}^{2 j_2} \mathcal{C}_{j_1, j_2, j, k_1, k_2} \left({O}_{1, k_1}{O}_{2, k_2}\right) u^{k_1+k_2+j-j_1-j_2},
\end{equation}
where the coefficients $\mathcal{C}$ are
\ie
&\mathcal{C}_{j_1, j_2, j, k_1, k_2}
\\
&=\sum_{r=0}^{j_1+j_2-j} \frac{(-)^r\left(j_1-j_2+j+r\right)_r^{\downarrow}}{r !(2 j+r+1)_r^{\downarrow}} \frac{1}{\left(j_1+j_2-j-r\right) !} \sum_{s=0}^r\left(\begin{array}{c}
r \\
s
\end{array}\right)\left(k_1\right)_{j_1+j_2-j-s}^{\downarrow}\left(k_2\right)_s^{\downarrow},
\fe
and the descending Pochhammer symbol is $(x)_k^{\downarrow}=\prod_{i=0}^{k-1}(x-i)$.

We close this discussion with several remarks
\begin{itemize}
    \item Our computations are only for $N=2, 3, 4, 5$, but based on the patterns we anticipate the above general formula for all $N$.    
    \item All the OPE coefficients in the above formula are \emph{analytic} in $N$. 
    This is reminiscent of the interesting relation between the (non-supersymmetric) $\mathcal{W}_N$ algebra and the $\mathcal{W}_{\infty}$ algebra, whose continuous parameter is the analytic continuation of $N$.
    \item For the conjecture of \cite{Beem:2013sza} to be consistent, we need the multi-trace corrections to always be suppressed by higher powers of $1/N$. So far, as we see above this is always the case, providing evidence for the conjecture.
    \item One salient feature of the non-supersymmetric $\mathcal{W}_{\infty}$ and $\mathcal{W}_N$ algebras is that $\mathcal{W}_{\infty}$ unifies all $\mathcal{W}_N$ algebras. The former depends on two parameters $c$ and $N$, and exhibits the nontrivial triality property, furnishing a much richer structure than the latter. It is a natural question to ask whether the $\mathcal{N}=4$ $\mathcal{W}$-algebra has similar features. We hope to come back to this in the future.
\end{itemize}

\subsection{Bootstrapping Schur black holes: Modular differential equations}
\label{Sec:Bootstrap}

Modular differential equations (MDEs) arose in the study of 2d CFT through two lines of thought---first as a consequence of null-vectors and Ward identities \cite{Eguchi:1986sb, Anderson:1987ge}, and second as a consequence of meromorphy and rationality \cite{Mathur:1988gt, Mathur:1988na}.  The two perspectives were unified in \cite{Gaberdiel:2008pr} in the rational case.  In the context of 4d $\cN = 2$ theories, it was realized \cite{Beem:2017ooy} through the discovery \cite{Beem:2013sza} of 2d chiral algebras that Schur indices obey MDEs; this observation was later generalized \cite{Pan:2021ulr} to flavored MDEs for flavored Schur indices.
We will use MDEs for (unflavored) Schur indices, which are (unflavored) vacuum characters of chiral algebras, to bootstrap a bound on the lightest $\frac18$-BPS black hole if it exists.

To simplify the discussion, in the following, we focus on odd rank $N$ and only mention the generalization to even $N$ toward the end.

\paragraph{Review of MDEs} 
We present a brief introduction to MDEs following the notation of \cite{Beem:2017ooy} and restricting ourselves to the modular group $\Gamma = \mathrm{PSL}(2,\bZ)$, which applies to $G=\SU(N)$ $\cN=4$ SYM with odd rank $N$.
A modular differential operator $\cD$ of order $d$ takes the form
\ie 
    \cD_q^{(d)} = \sum_{k=0}^d f_k(q) D_q^{(d-k)},
\fe 
where
\ie
    D_q^{(k)} := \partial_{(2k-2)} \circ \partial_{(2k-4)} \circ \dotsb \circ \partial_{(0)}
\fe 
and
\ie 
    \partial_{(k)} := q \frac{d}{dq} + k \mathbb{E}_2(q)
\fe
is the Serre derivative which turns a modular form of weight $k$ into one of weight $k+2$.\footnote{The precise definition of the Eisenstein series $\mathbb{E}_{2k}$ is not important for understanding our arguments.}
Each coefficient function $f_k(q)$ is a weakly holomorphic modular form of weight $2k$.  An MDE is called \emph{monic} if $f^{(0)}$ = 1.  It is well-known that weakly holomorphic modular forms of a fixed weight reside in a finite-dimensional vector space,\footnote{The ring of weakly holomorphic $\mathrm{PSL}(2,\bZ)$ modular forms is generated by $\mathbb{E}_4$ and $\mathbb{E}_6$.} and therefore, the space of modular differential operators of a fixed order $d$ is finite-dimensional.  
Here comes a key point: Since the null vector is a property of the chiral algebra, the character of any module must also satisfy the same MDE. 

\paragraph{MDEs for $\cN=4$ SYM} By considering the general ansatz of monic MDEs and explicitly acting them on the Schur indices of $\cN = 4$ SYM up to rank $N=7$, \cite{Beem:2017ooy} observed that these Schur indices satisfy monic MDEs of order $d = \left( \frac{N+1}{2} \right)^2$, not monic MDEs of any smaller order, and this degree formula was further supported in \cite{Beem:2021zvt}.  Moreover, the MDEs they found did not admit any solution that could serve as a candidate (unflavored) character for a chiral algebra module of positive weight.  It will be assumed that these observations hold for all $N$, and we have explicitly checked odd $N$ up to $N=23$.\footnote{We believe that a proof for all $N$ is possible using the closed-form formulae \cite{Bourdier:2015wda} for Schur indices, but we did not pursue this. The Mathematica file \texttt{mde.m} used for computing the MDEs and solving the indicial equations for odd $N$ is publicly available on \href{https://github.com/yinhslin/bps-counting/}{\texttt{https://github.com/yinhslin/bps-counting/}}.}
They further explained how a monic MDE could be constructed from a null vector of the form $(L_{-2}^d + \dotsc)\ket1$, which we review in Appendix~\ref{app:NSE}, but to our knowledge, it is unknown whether this MDE is the one of minimal order.

\paragraph{Bounding black holes} 
Suppose Conjecture~\ref{Conj} $=$\cite[Conjecture~3]{Beem:2013sza} were false, i.e.\ $\frac18$-BPS non-graviton operators exist, and suppose that the graviton operators do close as a chiral algebra. 
Then $\frac18$-BPS non-graviton operators, which we simply call black hole operators, 
must reside in positive-weight modules of the graviton super $\cW$-algebra.
We also assume the result of \cite{Gaberdiel:2008pr}, that MDEs must come from null vectors, holds beyond the rational (lisse) case to all quasi-lisse \cite{Arakawa:2016hkg} chiral algebras.
Under the above premises, there are the following two logical possibilities:
\begin{enumerate}
    \item The null vector giving rise to the MDE of minimal order $d = \left(\frac{N+1}{2}\right)^2$ is not contained in the graviton super $\cW$-algebra.  It must then contain a composite of a black hole operator with another operator; it must be a composite because otherwise, the null vector implies that the black hole operator is part of the graviton super $\cW$-algebra, violating its definition.  Since the lightest generator of the graviton super $\cW$-algebra has weight 1, we deduce that a black hole operator exists at level $h \le 2d-1$.
    \item The null vector is contained in the graviton super $\cW$-algebra. The unflavored character of the black hole module must vanish due to the lack of positive-weight solutions to the MDE. Such a cancellation suggests an exotic ghost parity in the black hole modules.
\end{enumerate}

Since we do not know any explanation of a ghost parity, we interpret the above logical exercise as suggesting that if black holes exist for odd rank $N$, they should be found at level $h \le 2d-1 = \frac{(N+1)^2}{2} - 1$.  
For even rank, one considers MDEs for modular group $\Gamma^0(2)$, and the observed minimal order is $d = \lfloor \left(\frac{N+1}{2}\right)^2 \rfloor = \frac{N(N+2)}{4}$.  

To summarize, 
the lightest $\frac18$-BPS black holes if existent in $\cN=4$ SYM of rank $N$ should have 
\ie\label{BootstrapBound}
    h \le
    \lfloor \frac{(N+1)^2}{2} \rfloor - 1,
\fe 
which provides a natural target and cutoff in our brute-force search of $\frac18$-BPS black hole operators.
The reader is cautioned that the above bound relied on several reasonable, yet not rigorously proven, assumptions.

\subsection{Brute-force construction and the graviton index method}
\label{sec:brute}

In \cite{Chang:2013fba}, a criterion for distinguishing $\frac{1}{16}$-BPS gravitons from non-gravitons as operators in $\cN=4$ SYM was formulated, and an algorithm to search for non-gravitons was proposed.  Upon improvements of the idea and algorithm, \cite{Chang:2022mjp} constructed the first non-graviton operator.  A shortcut to establishing the existence of non-gravitons was then proposed by \cite{Choi:2023znd}, which is to conduct the same brute-force search as in \cite{Chang:2022mjp} but restricted to gravitons for smaller computational complexity, and then compare the resulting graviton index with the full index to identify discrepancies.  
While a discrepancy implies the definite existence of a non-graviton, an agreement does not logically imply the non-existence of non-gravitons.

As explained in Section~\ref{Sec:Q}, the characterization of $\frac18$-BPS non-gravitons compared to $\frac{1}{16}$-BPS only requires imposing the charge relation $J_1 = J_2 = -q_3$, or equivalently, replacing the BPS superfield by the Schur superfield \eqref{eqn:Schur_superfield}.  The computational framework of \cite{Chang:2022mjp} is therefore immediately applicable to a $\frac18$-BPS brute-force search and the graviton index method.
A comprehensive search was performed for gauge groups $\SU(2)$, $\SU(3)$, $\SU(4)$ up to the levels $2h$ specified in Table~\ref{tab:n}, and no non-graviton was found, in support of Conjecture~\ref{Conj}.\footnote{The fully-refined enumeration data can be accessed on \href{https://github.com/yinhslin/bps-counting/}{\texttt{https://github.com/yinhslin/bps-counting/}}.
}

In light of the bootstrap arguments of Section~\ref{Sec:Bootstrap} and the bound \eqref{BootstrapBound}, which evaluates to
\ie 
    2h \le 6, 14, 22, \quad\text{for}\quad N = 2, 3, 4,
\fe 
it is plausible that no non-gravitons exist for arbitrarily large $h$ for $\SU(2)$ and $\SU(3)$.  For $\SU(3)$, the brute-force operator construction falls short of $2h = 14$, but the graviton index method is in strong support of the absence of non-gravitons. For $\SU(4)$, unfortunately, we do not have the resources to 
compute the graviton index up to $2h = 22$.

\begin{table}[t]
    \centering
    \begin{tabular}{|c|c|c|}
        \hline
        $N$ & All operators & Graviton operators
        \\\hline\hline
        2 & 16 & 24
        \\
        3 & 12 & 20
        \\
        4 & 10 & 16
        \\
        \hline
    \end{tabular}
    \caption{The maximal $2h$ (power of $x$ in SCI) of our comprehensive cohomological construction for $\SU(N)$ gauge groups. }
    \label{tab:n}
\end{table}

\section{Concluding remarks}
\label{Sec:Remarks}

Enticed by the prospects of capturing black hole dynamics in a structure as rigid as a chiral algebra, we embarked on a series of explorations for their existence in the $\cN=4$ SYM.  A priori, there are three scenarios, in decreasing order of optimism:
\begin{enumerate}
    \item Conjecture~\ref{Conj} is false, and the degeneracies of $\frac18$-BPS states exhibit black hole entropy growth, i.e.\ the black holes are large with macroscopic horizons.
    \item Conjecture~\ref{Conj} is false, but the degeneracies of $\frac18$-BPS states do not exhibit black hole entropy growth, i.e.\ the non-gravitons are not black holes, or the black holes are small without macroscopic horizons.
    \item  Conjecture~\ref{Conj} is true, that all $\frac18$-BPS states are gravitons, and no $\frac18$-BPS black hole or non-graviton object exists whatsoever.
\end{enumerate}
The various pieces of evidence we gathered could be summarized as follows:
\begin{itemize}
    \item We evaluated the large $N$ limit of the MacDonald index by the saddle point method. While the canonical index exhibited log-quadratic growth in $N$, the microcanonical index did not. This was against 1.
    \item The microcanonical flavored Schur index was rigorously proven to exhibit log-linear growth in $N$, which was also against Scenario~1.
    \item Considerations of modular differential equations suggested that if non-gravitons were to exist, i.e.\ Scenario~1 or 2 is true, then the lightest non-graviton should be bounded by \eqref{BootstrapBound}.
    \item A computerized construction of the full BPS cohomologies in the Schur sector agreed with the graviton cohomologies for $N=2$ beyond \eqref{BootstrapBound} and for $N=3$ close to \eqref{BootstrapBound}. The graviton index computed from the explicit graviton cohomologies agreed with the full index for $N=2,3,4$ up to the highest levels reached. Both strongly favored Scenario~3.
    \item An explicit evaluation of various microcanonical indices showed that the flavored MacDonald index grows at a similar pace as the flavored Schur index, both log-linear in $N$, while more slowly than a specialized SCI, which exhibited log-quadratic growth indicative of $\frac{1}{16}$-BPS black holes.
\end{itemize}

Overall, we find overwhelming evidence against Scenario~1 and strong evidence for Scenario~3.  The evidence against Scenario~1 suggests that $\frac18$-BPS AdS black hole solutions do not exist in supergravity, which is consistent with our long-standing failure to construct such solutions despite extensive efforts.
Even if such solutions existed, our results suggest that they should be unstable or lifted due to higher derivative corrections and quantum effects.  However, non-BPS black hole solutions in the Schur charge sector (relaxing the constraint on $D$ to an inequality) have recently been constructed \cite{Dias:2022eyq}, and it would be interesting to investigate the properties of near-BPS black holes in this charge sector.  On the field theory side, one could compute the anomalous dimensions \emph{en masse} following \cite{Chang:2023zqk} restricted to the Schur charge sector.  

Our evidence for Scenario~3, i.e.\ Conjecture~\ref{Conj} = \cite[Conjecture~3]{Beem:2013sza} of Beem and Rastelli, is strong but not fully conclusive.  It is logically possible that non-graviton modules exist but for mysterious reasons have vanishing (flavored) characters, or that the modular differential equations are accidental and do not arise from null vectors, thereby invalidating the bootstrap bound \eqref{BootstrapBound}.  It is also possible that while non-gravitons are absent for $N=2,3$, they appear for larger values of $N$.  Furthermore, it would be worth investigating whether the flavored MacDonald index admits alternative closed-form formulae, or if it could be majorized by a series that does, and rigorously prove the absence of black hole entropy growth in its coefficients.

What about black holes in other holographic $\cN = 2$ theories?  The holographic condition of $a_{\rm 4d} = c_{\rm 4d}$ at large $N$ on the 4d anomaly coefficients dictates \cite{Beem:2013sza,Beem:2017ooy} that the effective 2d central charge $c_{\rm 2d, eff} := c_{\rm 2d} - 24 \min_i(h_i) = 48 (c_{\rm 4d} - a_{\rm 4d}) = o(N^2)$.
When $h$ is scaled as $h = \alpha N^2$, the Cardy formula has a leading exponent proportional to $\sqrt{c_{\rm 2d, eff} h} = o(N^2)$, and hence the microcanonical Schur index does not exhibit macroscopic black hole growth.
Nevertheless, there could still be non-gravitons, and this certainly happens if the internal manifold has cycles giving rise to flavor symmetries under which gravitons are neutral.\footnote{
    For instance, such is the case in the $\cN = 1$ theories of \cite{Klebanov:1998hh}. 
}
{Indeed, for a large class of $\mathcal{N}=1$ theories, there is a good notion of graviton cohomologies (products of single-trace cohomologies) at finite $N$, and the single-trace cohomologies can be readily evaluated at infinite $N$
thanks to the Loday-Quillen-Tsygan theorem \cite{Budzik:2023xbr}, so it would be of great interest to investigate the nature of non-graviton operators.}
The formalism of the holomorphic twist \cite{Costello:2018zrm, Budzik:2023xbr} might be a promising step forward toward bringing black holes under strong theoretical control.

\section*{Acknowledgments}

We thank Christopher Beem, Pieter Bomans, Federico Bonetti, Sunjin Choi, Davide Gaiotto, Prahar Mitra, Yiwen Pan, Tom\'{a}\v{s} Proch\'{a}zka, Wenbin Yan, Xi Yin, and Keyou Zeng for the inspiring discussions. 
JW is grateful to the organizers of the Pollica Summer Workshop supported by the Regione Campania, Università degli Studi di Salerno, Università degli Studi di Napoli ``Federico II'', the Physics Department ``Ettore Pancini'' and ``E.R. Caianiello'', and Istituto Nazionale di Fisica Nucleare. CC is partly supported by National Key R\&D Program of China (NO. 2020YFA0713000). YL is supported by the Simons Collaboration Grant on the Non-Perturbative Bootstrap. The work of JW is supported by the UKRI Frontier Research Grant, underwriting the ERC Advanced Grant ``Generalized Symmetries in Quantum Field Theory and Quantum Gravity'' and the Simons Foundation Collaboration on
``Special Holonomy in Geometry, Analysis, and Physics'', Award ID: 724073, Schafer-Nameki.

\appendix

\section{An assortment of $\frac18$-BPS sectors}
\label{Sec:BPS}

In ${\cal N}=4$ SYM, there are three different $\frac18$-BPS sectors: the $\su(2|3)$ sector, the $\su(1,2|2)$ sector, and the $\su(1,1|2)$ sector. They preserve different sets of supercharges. Let us review the BPS conditions satisfied by the operators in these sectors and the corresponding BPS letters. 

We follow the convention in \cite{Chang:2023zqk} for the ${\cal N}=4$ superconformal algebra and the fundamental fields of ${\cal N}=4$ SYM. The $\frac{1}{16}$-BPS operators in ${\cal N}=4$ SYM are defined to be the operators annihilated by a single supercharge $Q^4_-$ and its Hermitian conjugates (BPZ conjugates) $Q_-^{4\dagger} = S_4^-$. The anti-commutator of them gives the BPS condition
\ie\label{eqn:1/16_comm}
&2\{Q^4_-,S_4^-\}=D-2(J_L)^{-}_{ -}-2R^4_4=D-J_1-J_2-q_1-q_2-q_3=0\,.
\fe
The $\frac{1}{16}$-BPS letters are the fundamental fields and covariant derivatives satisfying the BPS condition \eqref{eqn:1/16_comm},
\ie
&\phi^i=\Phi^{4i}\,,\quad \psi_i=-i\Psi_{+i}\,,\quad \lambda_{\dot\A}=\overline\Psi^4_{\dot\A}\,,
\\
&f=-iF_{++}=-i(\sigma^{\mu\nu})_{++}F_{\mu\nu}\,,\quad D_{\dot\A}=D_{+\dot\A}=\frac{1}{4}(\sigma^\mu)_{+\dot\A}D_\mu\,.
\fe

The $\frac18$-BPS operators are defined to be the operators that are annihilated by $Q^4_-$, $S_4^-$, an additional supercharge, and its Hermitian conjugate. Depending on the choice of the additional supercharge, there are different types of $\frac18$-BPS operators. Up to R-symmetry and Lorentz symmetry rotations, there are five inequivalent choices of the additional supercharge,
\ie
\overline Q_{4\dot -}\,,\quad Q^3_+\,,\quad Q^4_+\,,\quad Q^3_-\,,\quad \overline Q_{3\dot -}\,.
\fe
Let us analyze them one by one.

\begin{enumerate}

\item Consider the operators annihilated by
\ie
Q^4_-\,,\quad Q_-^{4\dagger} = S_4^-\,,\quad \overline Q_{4\dot -}\,,\quad \overline Q_{4\dot -}^\dagger = \overline S^{4\dot -}\,.
\fe
Because of the commutator
\ie
\{Q^4_-, \overline Q_{4\dot -}\}=P_{-\dot -}\,,
\fe
the operator must be annihilated by $P_{-\dot -}$. The only possibility of such an operator is the identity operator.

\item
Consider the operators annihilated by
\ie
Q^4_-\,,\quad Q_-^{4\dagger} = S_4^-\,,\quad  Q^3_+\,,\quad  Q^{3\dagger}_+ =  S^+_3\,.
\fe
Besides \eqref{eqn:1/16_comm}, there is only one extra nontrivial anti-commutator, which gives a new BPS condition
\ie
\{Q^3_+,S^+_3\}&=D+J_1+J_2-2R^3_3=D+J_1+J_2+q_1+q_2-q_3=0\,.
\fe
There is only one BPS letter
\ie
\phi^3\,.
\fe
Hence, the operators in this sector are given by products of traces of $\phi^3$'s and are actually half-BPS.

\item Consider the operators annihilated by
\ie
Q^4_-\,,\quad Q_-^{4\dagger} = S_4^-\,,\quad  Q^4_+\,,\quad  Q^{4\dagger}_+ =  S^+_4\,.
\fe
Besides \eqref{eqn:1/16_comm}, the nontrivial commutators are
\ie
\{Q^4_+,S^+_4\}&=D-2(J_L)^+_+-2R^4_4=D-2(J_L)^+_+-q_1-q_2-q_3\,,
\\
\{Q^4_+,S^-_4\}&=-2(J_L)^-_+\,,
\\
\{Q^4_-,S^+_4\}&=-2(J_L)^+_-\,,
\fe
where $(J_L)^\A_\B$ are the generators of the ${\rm SU}(2)_L\subset {\rm SO}(4)$.
The BPS conditions are
\ie
D=q_1+q_2+q_3\,,\quad  J_1+J_2=0\,.
\fe
The BPS letters are
\ie
\lambda_{\dot \alpha}\,,\quad \phi^i\,.
\fe
The space of the BPS words is called the $\mathfrak{su}(2|3)$ sector.

\item Consider the operators annihilated by
\ie
Q^4_-\,,\quad Q_-^{4\dagger} = S_4^-\,,\quad  Q^3_-\,,\quad  Q^{3\dagger}_- =  S^-_3\,.
\fe
Besides \eqref{eqn:1/16_comm}, the nontrivial commutators are
\ie
\{Q^3_-,S^-_3\}&=D-J_1-J_2-2R^3_3=D-J_1-J_2+q_1+q_2-q_3\,,
\\
\{Q^4_-,S^-_3\}&=-2R^4_3\,,
\\
\{Q^3_-,S^-_4\}&=-2R^3_4\,.
\fe
Hence, the BPS conditions are
\ie
D=J_1+J_2+q_3\,,\quad q_1+q_2=0\,.
\fe
The BPS letters are
\ie
\phi^3\,,\quad \psi_1\,,\quad \psi_2\,,\quad  f\,,\quad D_{\dot \alpha}\,.
\fe
The space of the BPS words is called the $\mathfrak{su}(1,2|2)$ sector.

\item Finally, consider the operators annihilated by
\ie
Q^4_-\,,\quad 
Q_-^{4\dagger} = S_4^-\,,\quad\overline Q_{\dot-3}\,,\quad  \overline Q_{\dot-3}^\dagger =\overline S^{\dot-3}\,.
\fe
This sector was discussed in Section~\ref{Sec:Q}, and it is called the $\mathfrak{su}(1,1|2)$ sector.

\end{enumerate}

\section{Null state equations from Zhu's $C_2$-algebra}
\label{app:NSE}

In \cite{Beem:2017ooy}, the authors argued that the chiral algebras of ${\cal N}=2$ superconformal field theories must admit a null state of the form $(L_{-2}^d+\cdots)\ket1$, where $\ket 1$ is the vacuum state and $d\in\bZ
_{>1}$. Let us briefly review the argument. 

Let ${\cal V}$ be the underlying vector space of a chiral algebra. For two currents $a,b\in{\cal V}$, the regularized product $(ab)$ and the secondary bracket $\{a,b\}$ are defined by the contour integrals
\ie\label{eqn:reg_2_product}
(ab)(z) & := \oint_{{\cal C}_1} \frac{dz_1}{2\pi i}\oint_{{\cal C}_2} \frac{dz_2}{2\pi i}\frac{a(z_1)b(z_2)}{(z_1-z)( z_2-z)}\,,
\\
\{a,b\}(z) & := \oint_{{\cal C}} \frac{dz'}{2\pi i}a(z')b(z)\,,
\fe
where ${\cal C}_1$ and ${\cal C}_2$ are counterclockwise  concentric circles around $z$ with the radius of ${\cal C}_1$ larger than that of ${\cal C}_2$, and ${\cal C}$ is a counterclockwise circle contour around $z$. The secondary bracket acts as a derivation on the regularized product
\ie
\{a,(bc)\}=(\{a,b\}c)+(b\{a,c\})\,.\label{eq:app:derivation}
\fe

The regularized product is not commutative and associative. The violations of them are captured by the commutator and the associator,
\ie
\relax [a,b]& :=(ab) - (ba)\,,
\\
[a,b,c]& :=((ab)c) - (a(bc))\,.
\fe
The secondary bracket is not a Lie bracket. The violations of it being a Lie bracket are captured by the symmetrizer and the jacobiator
\ie
\{a,b\}& :=\{a,b\}+\{b,a\}\,,
\\
\{a,b,c\}& :=\{a,\{b,c\}\}+\{b,\{c,a\}\}+\{c,\{a,b\}\}\,.
\fe
It is not hard to see that the commutator, associator, symmetrizer, and jacobiator are always inside a subspace $C_2({\cal V})$ in ${\cal V}$, which is roughly the space of regularized composite operators with at least one derivative. The Zhu's $C_2$ algebra ${\cal R}_{\cal V}$ of a chiral algebra is defined by the quotient
\ie
{\cal R}_{\cal V} := {\cal V} / C_2({\cal V})\,,
\fe
in which the regularized product becomes a commutative and associative product and the secondary bracket becomes a Lie bracket. In other words, the $C_2$ algebra ${\cal R}_{\cal V}$ is a commutative Poisson algebra.

There is another commutative Poisson algebra that shows up in a ${\cal N}=2$ superconformal field theory, namely the Higgs chiral ring ${\cal R}_H$. The main difference between ${\cal R}_H$ and ${\cal R}_{\cal V}$ is that ${\cal R}_H$ is reductive, i.e. it has no nonzero nilpotent elements. Based on these observations, the authors of \cite{Beem:2017ooy} conjectured that the Higgs chiral ring ${\cal R}_H$ is equal to the quotient of ${\cal R}_{\cal V}$ by its nilradical,
\ie
{\cal R}_H=({\cal R}_{\cal V})_{\rm red}\,.
\fe
This conjecture implies that every element inside ${\cal R}_{\cal V}$ but not ${\cal R}_H$ must be nilpotent in ${\cal R}_{\cal V}$, and hence raising it to a high enough power must give an element in $C_2({\cal V})$. In particular, the stress tensor $T$ is one of such elements, so for some positive integer $d$, we have $:T^d: = (T(T\cdots (TT)))$ equals to some regularized composite operators with at least one derivative. This gives the desired null state $(L_{-2}^d+\cdots)\ket1$.

\section{3-Laman graph and Feynman diagrams}\label{app:laman}

A graph $\Gamma$ is a pair $\Gamma=(\Gamma_0, \Gamma_1)$, where $\Gamma_0$ is the set of vertices, and $\Gamma_1$ is the set of edges specified by pairs of vertices. Let $S \subseteq \Gamma_0$ be any subset of the vertices of $\Gamma$. Then the induced subgraph $\Gamma[S]$ is the graph whose vertex set is $S$ and whose edge set consists of all of the edges in $\Gamma_1$ that have both endpoints in $S$. 

A graph $\Gamma$ is called $3$-Laman if
\begin{equation}
    3 |\Gamma_0| = 2|\Gamma_1|+4\,,
\end{equation}
and for every induced subgraph $S$,
\begin{equation}
    3 |\Gamma_0[S]| \geq  2 |\Gamma_1[S]| + 4.
\end{equation}
Figure~\ref{fig:3laman} lists the ten 3-Laman graphs with at most $8$ vertices.\footnote{It is a longstanding open question in graph theory how to efficiently enumerate all $n$-Laman graphs. The 2-Laman graphs can be generated recursively by the Henneberg moves. It is unknown if analogs of the Henneberg moves exist for general $n$-Laman graphs.}
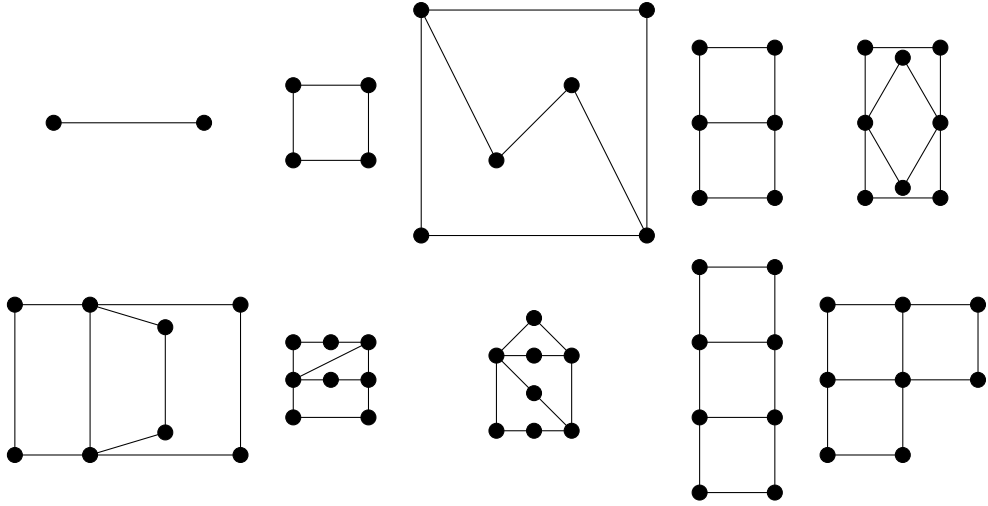
\begin{figure}[htb]
\centering
\setlength\tabcolsep{4pt} 
\renewcommand{\arraystretch}{3} 
\begin{tabular}{ccccc}
   \begin{tikzpicture}
	[
	baseline={(current bounding box.center)},
	line join=round
	]

	\coordinate (pd1) at (1.*\gS,0.*\gS);
	\coordinate (pd2) at (-1.*\gS,0.*\gS);

	\draw (pd1) node[GraphNode] {} node[left] {};
	\draw (pd2) node[GraphNode] {} node[left] {};

	\draw[GraphEdge] (pd1) -- (pd2) node[midway, above] {};

\end{tikzpicture}
&
\begin{tikzpicture}
	[
	baseline={(current bounding box.center)},
	line join=round
	]

	\coordinate (pd1) at (1.*\gS,2.*\gS);
	\coordinate (pd2) at (1.*\gS,1.*\gS);
	\coordinate (pd3) at (2.*\gS,1.*\gS);
	\coordinate (pd4) at (2.*\gS,2.*\gS);

	\draw (pd1) node[GraphNode] {} node[left] {};
	\draw (pd2) node[GraphNode] {} node[left] {};
	\draw (pd3) node[GraphNode] {} node[left] {};
	\draw (pd4) node[GraphNode] {} node[left] {};

	\draw[GraphEdge] (pd1) -- (pd2) node[midway, above] {};
	\draw[GraphEdge] (pd1) -- (pd4) node[midway, above] {};
	\draw[GraphEdge] (pd2) -- (pd3) node[midway, above] {};
	\draw[GraphEdge] (pd3) -- (pd4) node[midway, above] {};

\end{tikzpicture}
&
\begin{tikzpicture}
	[
	baseline={(current bounding box.center)},
	line join=round
	]

	\coordinate (pd1) at (3.*\gS,3.*\gS);
	\coordinate (pd2) at (1.*\gS,1.*\gS);
	\coordinate (pd3) at (4.*\gS,4.*\gS);
	\coordinate (pd4) at (2.*\gS,2.*\gS);
	\coordinate (pd5) at (4.*\gS,1.*\gS);
	\coordinate (pd6) at (1.*\gS,4.*\gS);

	\draw (pd1) node[GraphNode] {} node[left] {};
	\draw (pd2) node[GraphNode] {} node[left] {};
	\draw (pd3) node[GraphNode] {} node[left] {};
	\draw (pd4) node[GraphNode] {} node[left] {};
	\draw (pd5) node[GraphNode] {} node[left] {};
	\draw (pd6) node[GraphNode] {} node[left] {};

	\draw[GraphEdge] (pd1) -- (pd4) node[midway, above] {};
	\draw[GraphEdge] (pd1) -- (pd5) node[midway, above] {};
	\draw[GraphEdge] (pd2) -- (pd5) node[midway, above] {};
	\draw[GraphEdge] (pd2) -- (pd6) node[midway, above] {};
	\draw[GraphEdge] (pd3) -- (pd5) node[midway, above] {};
	\draw[GraphEdge] (pd3) -- (pd6) node[midway, above] {};
	\draw[GraphEdge] (pd4) -- (pd6) node[midway, above] {};

\end{tikzpicture}
&
\begin{tikzpicture}
	[
	baseline={(current bounding box.center)},
	line join=round
	]

	\coordinate (pd1) at (1.*\gS,1.*\gS);
	\coordinate (pd2) at (2.*\gS,1.*\gS);
	\coordinate (pd3) at (1.*\gS,2.*\gS);
	\coordinate (pd4) at (2.*\gS,2.*\gS);
	\coordinate (pd5) at (1.*\gS,3.*\gS);
	\coordinate (pd6) at (2.*\gS,3.*\gS);

	\draw (pd1) node[GraphNode] {} node[left] {};
	\draw (pd2) node[GraphNode] {} node[left] {};
	\draw (pd3) node[GraphNode] {} node[left] {};
	\draw (pd4) node[GraphNode] {} node[left] {};
	\draw (pd5) node[GraphNode] {} node[left] {};
	\draw (pd6) node[GraphNode] {} node[left] {};

	\draw[GraphEdge] (pd1) -- (pd2) node[midway, above] {};
	\draw[GraphEdge] (pd1) -- (pd3) node[midway, above] {};
	\draw[GraphEdge] (pd2) -- (pd4) node[midway, above] {};
	\draw[GraphEdge] (pd3) -- (pd4) node[midway, above] {};
	\draw[GraphEdge] (pd3) -- (pd5) node[midway, above] {};
	\draw[GraphEdge] (pd4) -- (pd6) node[midway, above] {};
	\draw[GraphEdge] (pd5) -- (pd6) node[midway, above] {};

\end{tikzpicture}
&
\begin{tikzpicture}
	[
	baseline={(current bounding box.center)},
	line join=round
	]

	\coordinate (pd1) at (0.5*\gS,1.*\gS);
	\coordinate (pd2) at (0.5*\gS,-1.*\gS);
	\coordinate (pd3) at (0.*\gS,0.866*\gS);
	\coordinate (pd4) at (0.*\gS,-0.866*\gS);
	\coordinate (pd5) at (-0.5*\gS,1.*\gS);
	\coordinate (pd6) at (-0.5*\gS,-1.*\gS);
	\coordinate (pd7) at (0.5*\gS,0.*\gS);
	\coordinate (pd8) at (-0.5*\gS,0.*\gS);

	\draw (pd1) node[GraphNode] {} node[left] {};
	\draw (pd2) node[GraphNode] {} node[left] {};
	\draw (pd3) node[GraphNode] {} node[left] {};
	\draw (pd4) node[GraphNode] {} node[left] {};
	\draw (pd5) node[GraphNode] {} node[left] {};
	\draw (pd6) node[GraphNode] {} node[left] {};
	\draw (pd7) node[GraphNode] {} node[left] {};
	\draw (pd8) node[GraphNode] {} node[left] {};

	\draw[GraphEdge] (pd1) -- (pd5) node[midway, above] {};
	\draw[GraphEdge] (pd1) -- (pd7) node[midway, above] {};
	\draw[GraphEdge] (pd2) -- (pd6) node[midway, above] {};
	\draw[GraphEdge] (pd2) -- (pd7) node[midway, above] {};
	\draw[GraphEdge] (pd3) -- (pd7) node[midway, above] {};
	\draw[GraphEdge] (pd3) -- (pd8) node[midway, above] {};
	\draw[GraphEdge] (pd4) -- (pd7) node[midway, above] {};
	\draw[GraphEdge] (pd4) -- (pd8) node[midway, above] {};
	\draw[GraphEdge] (pd5) -- (pd8) node[midway, above] {};
	\draw[GraphEdge] (pd6) -- (pd8) node[midway, above] {};

\end{tikzpicture}
\\
\begin{tikzpicture}
	[
	baseline={(current bounding box.center)},
	line join=round
	]

	\coordinate (pd1) at (-1.*\gS,1.*\gS);
	\coordinate (pd2) at (1.*\gS,0.7*\gS);
	\coordinate (pd3) at (2.*\gS,1.*\gS);
	\coordinate (pd4) at (0.*\gS,1.*\gS);
	\coordinate (pd5) at (-1.*\gS,-1.*\gS);
	\coordinate (pd6) at (1.*\gS,-0.7*\gS);
	\coordinate (pd7) at (2.*\gS,-1.*\gS);
	\coordinate (pd8) at (0.*\gS,-1.*\gS);

	\draw (pd1) node[GraphNode] {} node[left] {};
	\draw (pd2) node[GraphNode] {} node[left] {};
	\draw (pd3) node[GraphNode] {} node[left] {};
	\draw (pd4) node[GraphNode] {} node[left] {};
	\draw (pd5) node[GraphNode] {} node[left] {};
	\draw (pd6) node[GraphNode] {} node[left] {};
	\draw (pd7) node[GraphNode] {} node[left] {};
	\draw (pd8) node[GraphNode] {} node[left] {};

	\draw[GraphEdge] (pd1) -- (pd4) node[midway, above] {};
	\draw[GraphEdge] (pd1) -- (pd5) node[midway, above] {};
	\draw[GraphEdge] (pd2) -- (pd4) node[midway, above] {};
	\draw[GraphEdge] (pd2) -- (pd6) node[midway, above] {};
	\draw[GraphEdge] (pd3) -- (pd4) node[midway, above] {};
	\draw[GraphEdge] (pd3) -- (pd7) node[midway, above] {};
	\draw[GraphEdge] (pd4) -- (pd8) node[midway, above] {};
	\draw[GraphEdge] (pd5) -- (pd8) node[midway, above] {};
	\draw[GraphEdge] (pd6) -- (pd8) node[midway, above] {};
	\draw[GraphEdge] (pd7) -- (pd8) node[midway, above] {};

\end{tikzpicture}
&
\begin{tikzpicture}
	[
	baseline={(current bounding box.center)},
	line join=round
	]

	\coordinate (pd1) at (0.5*\gS,1.*\gS);
	\coordinate (pd2) at (1.*\gS,0.*\gS);
	\coordinate (pd3) at (0.*\gS,1.*\gS);
	\coordinate (pd4) at (0.*\gS,0.*\gS);
	\coordinate (pd5) at (0.5*\gS,0.5*\gS);
	\coordinate (pd6) at (1.*\gS,0.5*\gS);
	\coordinate (pd7) at (1.*\gS,1.*\gS);
	\coordinate (pd8) at (0.*\gS,0.5*\gS);

	\draw (pd1) node[GraphNode] {} node[left] {};
	\draw (pd2) node[GraphNode] {} node[left] {};
	\draw (pd3) node[GraphNode] {} node[left] {};
	\draw (pd4) node[GraphNode] {} node[left] {};
	\draw (pd5) node[GraphNode] {} node[left] {};
	\draw (pd6) node[GraphNode] {} node[left] {};
	\draw (pd7) node[GraphNode] {} node[left] {};
	\draw (pd8) node[GraphNode] {} node[left] {};

	\draw[GraphEdge] (pd1) -- (pd3) node[midway, above] {};
	\draw[GraphEdge] (pd1) -- (pd7) node[midway, above] {};
	\draw[GraphEdge] (pd2) -- (pd4) node[midway, above] {};
	\draw[GraphEdge] (pd2) -- (pd6) node[midway, above] {};
	\draw[GraphEdge] (pd3) -- (pd8) node[midway, above] {};
	\draw[GraphEdge] (pd4) -- (pd8) node[midway, above] {};
	\draw[GraphEdge] (pd5) -- (pd6) node[midway, above] {};
	\draw[GraphEdge] (pd5) -- (pd8) node[midway, above] {};
	\draw[GraphEdge] (pd6) -- (pd7) node[midway, above] {};
	\draw[GraphEdge] (pd7) -- (pd8) node[midway, above] {};

\end{tikzpicture}
&
\begin{tikzpicture}
	[
	baseline={(current bounding box.center)},
	line join=round
	]

	\coordinate (pd1) at (0.5*\gS,0.*\gS);
	\coordinate (pd2) at (0.5*\gS,0.5*\gS);
	\coordinate (pd3) at (0.5*\gS,1.*\gS);
	\coordinate (pd4) at (0.5*\gS,1.5*\gS);
	\coordinate (pd5) at (0.*\gS,0.*\gS);
	\coordinate (pd6) at (1.*\gS,1.*\gS);
	\coordinate (pd7) at (1.*\gS,0.*\gS);
	\coordinate (pd8) at (0.*\gS,1.*\gS);

	\draw (pd1) node[GraphNode] {} node[left] {};
	\draw (pd2) node[GraphNode] {} node[left] {};
	\draw (pd3) node[GraphNode] {} node[left] {};
	\draw (pd4) node[GraphNode] {} node[left] {};
	\draw (pd5) node[GraphNode] {} node[left] {};
	\draw (pd6) node[GraphNode] {} node[left] {};
	\draw (pd7) node[GraphNode] {} node[left] {};
	\draw (pd8) node[GraphNode] {} node[left] {};

	\draw[GraphEdge] (pd1) -- (pd5) node[midway, above] {};
	\draw[GraphEdge] (pd1) -- (pd7) node[midway, above] {};
	\draw[GraphEdge] (pd2) -- (pd7) node[midway, above] {};
	\draw[GraphEdge] (pd2) -- (pd8) node[midway, above] {};
	\draw[GraphEdge] (pd3) -- (pd6) node[midway, above] {};
	\draw[GraphEdge] (pd3) -- (pd8) node[midway, above] {};
	\draw[GraphEdge] (pd4) -- (pd6) node[midway, above] {};
	\draw[GraphEdge] (pd4) -- (pd8) node[midway, above] {};
	\draw[GraphEdge] (pd5) -- (pd8) node[midway, above] {};
	\draw[GraphEdge] (pd6) -- (pd7) node[midway, above] {};

\end{tikzpicture}
&
\begin{tikzpicture}
	[
	baseline={(current bounding box.center)},
	line join=round
	]

	\coordinate (pd1) at (1.*\gS,1.*\gS);
	\coordinate (pd2) at (2.*\gS,1.*\gS);
	\coordinate (pd3) at (1.*\gS,2.*\gS);
	\coordinate (pd4) at (2.*\gS,2.*\gS);
	\coordinate (pd5) at (1.*\gS,3.*\gS);
	\coordinate (pd6) at (2.*\gS,3.*\gS);
	\coordinate (pd7) at (1.*\gS,4.*\gS);
	\coordinate (pd8) at (2.*\gS,4.*\gS);

	\draw (pd1) node[GraphNode] {} node[left] {};
	\draw (pd2) node[GraphNode] {} node[left] {};
	\draw (pd3) node[GraphNode] {} node[left] {};
	\draw (pd4) node[GraphNode] {} node[left] {};
	\draw (pd5) node[GraphNode] {} node[left] {};
	\draw (pd6) node[GraphNode] {} node[left] {};
	\draw (pd7) node[GraphNode] {} node[left] {};
	\draw (pd8) node[GraphNode] {} node[left] {};

	\draw[GraphEdge] (pd1) -- (pd2) node[midway, above] {};
	\draw[GraphEdge] (pd1) -- (pd3) node[midway, above] {};
	\draw[GraphEdge] (pd2) -- (pd4) node[midway, above] {};
	\draw[GraphEdge] (pd3) -- (pd4) node[midway, above] {};
	\draw[GraphEdge] (pd3) -- (pd5) node[midway, above] {};
	\draw[GraphEdge] (pd4) -- (pd6) node[midway, above] {};
	\draw[GraphEdge] (pd5) -- (pd6) node[midway, above] {};
	\draw[GraphEdge] (pd5) -- (pd7) node[midway, above] {};
	\draw[GraphEdge] (pd6) -- (pd8) node[midway, above] {};
	\draw[GraphEdge] (pd7) -- (pd8) node[midway, above] {};

\end{tikzpicture}
&
\begin{tikzpicture}
	[
	baseline={(current bounding box.center)},
	line join=round
	]

	\coordinate (pd1) at (0.*\gS,0.*\gS);
	\coordinate (pd2) at (0.*\gS,-1.*\gS);
	\coordinate (pd3) at (0.*\gS,-2.*\gS);
	\coordinate (pd4) at (1.*\gS,0.*\gS);
	\coordinate (pd5) at (1.*\gS,-1.*\gS);
	\coordinate (pd6) at (1.*\gS,-2.*\gS);
	\coordinate (pd7) at (2.*\gS,0.*\gS);
	\coordinate (pd8) at (2.*\gS,-1.*\gS);

	\draw (pd1) node[GraphNode] {} node[left] {};
	\draw (pd2) node[GraphNode] {} node[left] {};
	\draw (pd3) node[GraphNode] {} node[left] {};
	\draw (pd4) node[GraphNode] {} node[left] {};
	\draw (pd5) node[GraphNode] {} node[left] {};
	\draw (pd6) node[GraphNode] {} node[left] {};
	\draw (pd7) node[GraphNode] {} node[left] {};
	\draw (pd8) node[GraphNode] {} node[left] {};

	\draw[GraphEdge] (pd1) -- (pd2) node[midway, above] {};
	\draw[GraphEdge] (pd1) -- (pd4) node[midway, above] {};
	\draw[GraphEdge] (pd2) -- (pd3) node[midway, above] {};
	\draw[GraphEdge] (pd2) -- (pd5) node[midway, above] {};
	\draw[GraphEdge] (pd3) -- (pd6) node[midway, above] {};
	\draw[GraphEdge] (pd4) -- (pd5) node[midway, above] {};
	\draw[GraphEdge] (pd4) -- (pd7) node[midway, above] {};
	\draw[GraphEdge] (pd5) -- (pd6) node[midway, above] {};
	\draw[GraphEdge] (pd5) -- (pd8) node[midway, above] {};
	\draw[GraphEdge] (pd7) -- (pd8) node[midway, above] {};

\end{tikzpicture}
\end{tabular}
    \caption{List of all 3-Laman graphs with at most $8$ vertices.}
    \label{fig:3laman}
\end{figure}

It was shown in \cite{Budzik:2022mpd} that Feynman integrals in the HT-twisted theory take the following form in the Schwinger formalism:
\begin{equation}
{\cal I}_\Gamma[\lambda;\delta] := \int_{\mathbb{R}^{4 |\Gamma_0|-4} \times \partial_\epsilon \left[[\epsilon,L]^{|\Gamma_1|}\right]} \left[ \prod_{v \in \Gamma_0|v \neq v_0} e^{\lambda_v  x_v} \frac{d z_v}{(2\pi i)} \right] \bar{\partial}\left[\prod_{e \in \Gamma_1} {\cal P}(x_{e,i}, z_{e} +\delta_e, \bar z_{e} ,t_e)\right] \,,
\end{equation}
where to each vertex $v \neq v_0$ we associate  a position $(x_{v,1},x_{v,2}, z_v, \bar{z}_v) \in \mathbb{R}^4$ and a holomorphic $\lambda$-parameter $\lambda_v \in \mathbb{C}$.
To each edge $e$ we associate a propagator, 
\begin{equation}
    {\cal P}(x_{e,i},z_e, \bar{z}_{e},t_e) = \frac{1}{\sqrt{\pi}} e^{- s_{e,1}^2 -s_{e,2}^2- z_e y_e}  dy_e ds_{e,1}ds_{e,2}\,, \qquad s_{e,i} = \frac{x_{e,i}}{\sqrt{t_e}}\,, \qquad y_e = \frac{\bar z_e}{t_e}\, . \label{eq:superpropagator}
\end{equation}
with  Schwinger time $t_e$ and an extra holomorphic shift $\delta_e \in \mathbb{C}$. The evaluation of the integral, although not impossible, is generally difficult. However, a simple form-degree counting shows that $\mathcal{I}_{\Gamma}$ is trivial unless $\Gamma$ is a $3$-Laman graph. For more details see \cite{Budzik:2022mpd}.

The integral over the spacetime is Gaussian and can be easily performed, leaving us with an integral 
over the Schwinger parameters $t_e$,
\begin{equation}
   {\cal I}_\Gamma[\lambda;\delta] := \int_{ \partial_\epsilon \left[[\epsilon,L]^{|\Gamma_1|}\right]} \omega_{\Gamma}.
\end{equation}
Furthermore, since \eqref{eq:superpropagator} is fully factorized into the holomorphic and topological coordinates, $\omega_{\Gamma}$ must also take the factorized form
\ie
\omega_{\Gamma}=\rho_{\Gamma}\left(\lambda_*^i\right) \wedge\left(\eta_{\Gamma}\right)^2,
\fe
where $\rho_{\Gamma}$ and $\eta_{\Gamma}$ denote the contributions from the holomorphic and topological directions, respectively. Computing the explicit expressions for $\omega_{\Gamma}$ is straightforward but tedious. With the help of Mathematica, one can easily compute $\omega_{\Gamma}$ for the $3$-Laman graphs with at most 8 vertices, listed in Figure~\ref{fig:3laman}, and interestingly, they all vanish. We conjecture that this is true for all 3-Laman graphs. It would be interesting to pursue a general combinatorial proof of this vanishing, and we plan to come back to this in the future.

\section{Tables of microcanonical indices}
\label{App:Tables}

\begin{table}[H]
    \centering
    \begin{tabular}{|c||c|c|cc|c|c|c|c|c|c|}
        \hline
        $2h$ & $\U(1)$ & $\U(2)$ & $\U(3)_\S$ & $\U(3)_\FS$ & $\U(4)$ & $\U(5)$ & $\U(6)$ & $\U(7)$ & $\U(8)$ & $\U(9)$ 
        \\\hline\hline
        0 & 1 & 1 & 1 & 1 & 1 & 1 & 1 & 1 & 1 & 1 \\
        1 & 2 & 2 & 2 & 2 & 2 & 2 & 2 & 2 & 2 & 2 \\
        2 & 1 & 4 & 4 & 4 & 4 & 4 & 4 & 4 & 4 & 4 \\
        3 & 2 & 4 & 8 & 8 & 8 & 8 & 8 & 8 & 8 & 8 \\
        4 & 2 & 6 & 9 & 9 & 14 & 14 & 14 & 14 & 14 & 14 \\
        5 & 0 & 8 & 14 & 14 & 18 & 24 & 24 & 24 & 24 & 24 \\
        6 & 3 & 8 & 20 & 20 & 28 & 33 & 40 & 40 & 40 & 40 \\
        7 & 2 & 8 & 24 & 24 & 40 & 50 & 56 & 64 & 64 & 64 \\
        8 & 0 & 13 & 30 & 30 & 52 & 72 & 84 & 91 & 100 & 100 \\
        9 & 2 & 12 & 34 & 34 & 70 & 98 & 122 & 136 & 144 & 154 \\
        10 & 2 & 12 & 46 & 46 & 88 & 134 & 168 & 196 & 212 & 221 \\
        11 & 2 & 16 & 52 & 52 & 104 & 176 & 232 & 272 & 304 & 322 \\
        12 & 1 & 14 & 60 & 60 & 140 & 224 & 312 & 378 & 424 & 460 \\
        13 & 2 & 16 & 70 & 70 & 168 & 280 & 408 & 512 & 588 & 640 \\
        14 & 0 & 24 & 76 & 76 & 196 & 367 & 528 & 680 & 800 & 886 \\
        15 & 2 & 16 & 88 & 88 & 240 & 448 & 672 & 896 & 1072 & 1208 \\
        16 & 4 & 18 & 102 & 102 & 278 & 546 & 865 & 1162 & 1422 & 1622 \\
        17 & 0 & 26 & 120 & 120 & 320 & 672 & 1078 & 1478 & 1864 & 2160 \\
        18 & 2 & 20 & 118 & 126 & 380 & 798 & 1336 & 1904 & 2408 & 2848 \\
        19 & 0 & 24 & 136 & 136 & 440 & 952 & 1648 & 2392 & 3080 & 3706 \\
        20 & 1 & 32 & 154 & 154 & 504 & 1136 & 2002 & 2976 & 3950 & 4784 \\
        21 & 4 & 24 & 160 & 160 & 562 & 1328 & 2424 & 3704 & 4972 & 6120 \\
        22 & 2 & 24 & 184 & 184 & 644 & 1554 & 2912 & 4544 & 6224 & 7809 \\
        23 & 0 & 32 & 200 & 200 & 720 & 1838 & 3472 & 5536 & 7760 & 9850 \\
        24 & 2 & 31 & 208 & 224 & 808 & 2079 & 4116 & 6730 & 9564 & 12348 \\
        25 & 2 & 28 & 234 & 234 & 910 & 2400 & 4872 & 8106 & 11742 & 15394 \\
        26 & 0 & 40 & 252 & 252 & 1000 & 2772 & 5744 & 9692 & 14344 & 19052 \\
        27 & 2 & 32 & 244 & 244 & 1120 & 3120 & 6648 & 11584 & 17384 & 23464 \\
        28 & 2 & 30 & 295 & 295 & 1240 & 3554 & 7752 & 13752 & 20968 & 28740 \\
        29 & 2 & 48 & 308 & 308 & 1360 & 4048 & 8976 & 16248 & 25204 & 35008 \\
        30 & 1 & 32 & 308 & 340 & 1488 & 4522 & 10304 & 19052 & 30112 & 42428 \\\hline
    \end{tabular}
    \caption{The absolute degeneracies of Schur indices for $\U(N)$ gauge groups with $N=1, \dotsc, 9$.  Up to $2h=30$, the unflavored and flavored indices have identical absolute degeneracies except for $\U(3)$.}
    \label{Tab:U}
\end{table}

\newpage 

\begin{table}[H]
    \centering
    \begin{tabular}{|c||cc|cc|cc|cc|}
        \hline
        $2h$ & $\SU(2)_\S$ & $\SU(2)_\FS$ & $\SU(4)_\S$ & $\SU(4)_\FS$ & $\SU(6)_\S$ & $\SU(6)_\FS$ & $\SU(8)_\S$ & $\SU(8)_\FS$ 
        \\\hline\hline
        0 & 1 & 1 & 1 & 1 & 1 & 1 & 1 & 1 \\
        1 & 0 & 0 & 0 & 0 & 0 & 0 & 0 & 0 \\
        2 & 3 & 3 & 3 & 3 & 3 & 3 & 3 & 3 \\
        3 & 4 & 4 & 0 & 4 & 0 & 4 & 0 & 4 \\
        4 & 9 & 9 & 9 & 9 & 9 & 9 & 9 & 9 \\
        5 & 12 & 12 & 6 & 10 & 0 & 8 & 0 & 8 \\
        6 & 22 & 22 & 22 & 22 & 22 & 22 & 22 & 22 \\
        7 & 36 & 36 & 18 & 26 & 8 & 20 & 0 & 16 \\
        8 & 60 & 60 & 51 & 51 & 51 & 51 & 51 & 51 \\
        9 & 88 & 88 & 54 & 66 & 24 & 48 & 10 & 38 \\
        10 & 135 & 135 & 108 & 108 & 108 & 108 & 108 & 108 \\
        11 & 204 & 204 & 132 & 148 & 72 & 108 & 30 & 78 \\
        12 & 302 & 302 & 241 & 241 & 221 & 221 & 221 & 221 \\
        13 & 432 & 432 & 306 & 326 & 176 & 228 & 90 & 166 \\
        14 & 627 & 627 & 489 & 489 & 429 & 429 & 429 & 429 \\
        15 & 900 & 900 & 648 & 676 & 408 & 488 & 220 & 340 \\
        16 & 1269 & 1269 & 990 & 990 & 845 & 845 & 810 & 810 \\
        17 & 1764 & 1764 & 1326 & 1358 & 864 & 972 & 510 & 686 \\
        18 & 2451 & 2451 & 1919 & 1919 & 1584 & 1584 & 1479 & 1479 \\
        19 & 3384 & 3384 & 2574 & 2614 & 1768 & 1912 & 1080 & 1336 \\
        20 & 4629 & 4629 & 3660 & 3660 & 2955 & 2955 & 2694 & 2694 \\
        21 & 6268 & 6268 & 4910 & 4958 & 3432 & 3628 & 2210 & 2574 \\
        22 & 8460 & 8460 & 6759 & 6759 & 5369 & 5369 & 4761 & 4761 \\
        23 & 11376 & 11376 & 9024 & 9080 & 6480 & 6732 & 4290 & 4806 \\
        24 & 15183 & 15183 & 12288 & 12288 & 9653 & 9653 & 8354 & 8354 \\
        25 & 20124 & 20124 & 16290 & 16354 & 11832 & 12152 & 8100 & 8812 \\
        26 & 26595 & 26595 & 21789 & 21789 & 16989 & 16989 & 14397 & 14397 \\
        27 & 35008 & 35008 & 28694 & 28770 & 21232 & 21640 & 14790 & 15750 \\
        28 & 45828 & 45828 & 38043 & 38043 & 29578 & 29578 & 24597 & 24597 \\
        29 & 59700 & 59700 & 49758 & 49842 & 37128 & 37632 & 26400 & 27668 \\
        30 & 77522 & 77522 & 65161 & 65161 & 50596 & 50596 & 41413 & 41413 \\\hline
    \end{tabular}
    \caption{The absolute degeneracies of unflavored and flavored Schur indices for $\SU(N)$ gauge groups with $N=2, 4, 6, 8$.}
    \label{Tab:SUEven}
\end{table}

\newpage 

\begin{table}[H]
    \centering
    \begin{tabular}{|c||cc|cc|cc|cc|}
        \hline
        $2h$ & $\SU(3)_\S$ & $\SU(3)_\FS$ & $\SU(5)_\S$ & $\SU(5)_\FS$ & $\SU(7)_\S$ & $\SU(7)_\FS$ & $\SU(9)_\S$ & $\SU(9)_\FS$ 
        \\\hline\hline
        0 & 1 & 1 & 1 & 1 & 1 & 1 & 1 & 1 \\
        1 & 0 & 0 & 0 & 0 & 0 & 0 & 0 & 0 \\
        2 & 3 & 3 & 3 & 3 & 3 & 3 & 3 & 3 \\
        3 & 0 & 4 & 0 & 4 & 0 & 4 & 0 & 4 \\
        4 & 4 & 4 & 9 & 9 & 9 & 9 & 9 & 9 \\
        5 & 0 & 8 & 0 & 8 & 0 & 8 & 0 & 8 \\
        6 & 7 & 11 & 15 & 15 & 22 & 22 & 22 & 22 \\
        7 & 0 & 16 & 0 & 16 & 0 & 16 & 0 & 16 \\
        8 & 6 & 22 & 30 & 30 & 42 & 42 & 51 & 51 \\
        9 & 0 & 32 & 0 & 28 & 0 & 32 & 0 & 32 \\
        10 & 12 & 40 & 45 & 45 & 81 & 81 & 97 & 97 \\
        11 & 0 & 64 & 0 & 52 & 0 & 60 & 0 & 64 \\
        12 & 8 & 88 & 67 & 79 & 140 & 140 & 188 & 188 \\
        13 & 0 & 120 & 0 & 88 & 0 & 112 & 0 & 120 \\
        14 & 15 & 163 & 99 & 143 & 231 & 231 & 330 & 330 \\
        15 & 0 & 228 & 0 & 160 & 0 & 192 & 0 & 224 \\
        16 & 13 & 309 & 135 & 243 & 351 & 351 & 568 & 568 \\
        17 & 0 & 420 & 0 & 272 & 0 & 328 & 0 & 388 \\
        18 & 18 & 562 & 175 & 407 & 551 & 563 & 918 & 918 \\
        19 & 0 & 748 & 0 & 480 & 0 & 528 & 0 & 664 \\
        20 & 12 & 992 & 231 & 679 & 783 & 815 & 1452 & 1452 \\
        21 & 0 & 1312 & 0 & 800 & 0 & 840 & 0 & 1104 \\
        22 & 28 & 1712 & 306 & 1114 & 1134 & 1286 & 2233 & 2233 \\
        23 & 0 & 2256 & 0 & 1356 & 0 & 1304 & 0 & 1728 \\
        24 & 14 & 2934 & 354 & 1818 & 1546 & 1934 & 3344 & 3344 \\
        25 & 0 & 3808 & 0 & 2224 & 0 & 2012 & 0 & 2720 \\
        26 & 24 & 4900 & 465 & 2961 & 2142 & 2958 & 4884 & 4884 \\
        27 & 0 & 6308 & 0 & 3652 & 0 & 3080 & 0 & 4144 \\
        28 & 24 & 8068 & 540 & 4764 & 2835 & 4431 & 7004 & 7212 \\
        29 & 0 & 10312 & 0 & 5872 & 0 & 4804 & 0 & 6220 \\
        30 & 31 & 13107 & 681 & 7597 & 3758 & 6526 & 9856 & 10420 \\\hline
    \end{tabular}
    \caption{The absolute degeneracies of unflavored and flavored Schur indices for $\SU(N)$ gauge groups with $N=3, 5, 7, 9$.}
    \label{Tab:SUOdd}
\end{table}

\newpage 

\begin{table}[H]
    \centering
    \begin{tabular}{|c||c|c|c|c|c|c|c|c|c|}
        \hline
        $2h$ & $\U(1)$ & $\U(2)$ & $\U(3)$ & $\U(4)$ & $\U(5)$ & $\U(6)$ & $\U(7)$ & $\U(8)$ & $\U(9)$ 
        \\\hline\hline
        0 & 1 & 1 & 1 & 1 & 1 & 1 & 1 & 1 & 1 \\
        1 & 2 & 2 & 2 & 2 & 2 & 2 & 2 & 2 & 2 \\
        2 & 3 & 6 & 6 & 6 & 6 & 6 & 6 & 6 & 6 \\
        3 & 2 & 8 & 12 & 12 & 12 & 12 & 12 & 12 & 12 \\
        4 & 4 & 12 & 21 & 26 & 26 & 26 & 26 & 26 & 26 \\
        5 & 4 & 16 & 30 & 42 & 48 & 48 & 48 & 48 & 48 \\
        6 & 3 & 24 & 42 & 68 & 83 & 90 & 90 & 90 & 90 \\
        7 & 6 & 24 & 60 & 96 & 130 & 148 & 156 & 156 & 156 \\
        8 & 4 & 31 & 82 & 142 & 198 & 240 & 261 & 270 & 270 \\
        9 & 6 & 40 & 102 & 186 & 290 & 362 & 412 & 436 & 446 \\
        10 & 8 & 48 & 124 & 260 & 410 & 548 & 636 & 694 & 721 \\
        11 & 6 & 52 & 164 & 332 & 564 & 780 & 952 & 1056 & 1122 \\
        12 & 9 & 68 & 204 & 426 & 766 & 1112 & 1390 & 1596 & 1716 \\
        13 & 6 & 64 & 238 & 536 & 1020 & 1524 & 1984 & 2324 & 2564 \\
        14 & 12 & 82 & 296 & 690 & 1319 & 2086 & 2786 & 3360 & 3762 \\
        15 & 10 & 96 & 340 & 832 & 1704 & 2780 & 3836 & 4732 & 5420 \\
        16 & 12 & 96 & 376 & 1040 & 2170 & 3667 & 5218 & 6608 & 7700 \\
        17 & 16 & 118 & 448 & 1236 & 2720 & 4754 & 7002 & 9044 & 10768 \\
        18 & 8 & 142 & 524 & 1470 & 3396 & 6150 & 9244 & 12294 & 14882 \\
        19 & 20 & 124 & 588 & 1736 & 4176 & 7796 & 12108 & 16440 & 20310 \\
        20 & 13 & 172 & 686 & 2058 & 5022 & 9874 & 15702 & 21800 & 27428 \\
        21 & 20 & 176 & 784 & 2398 & 6020 & 12296 & 20136 & 28572 & 36668 \\
        22 & 20 & 188 & 830 & 2862 & 7204 & 15176 & 25630 & 37210 & 48523 \\
        23 & 20 & 220 & 968 & 3304 & 8506 & 18424 & 32328 & 47888 & 63686 \\
        24 & 28 & 227 & 1096 & 3826 & 10157 & 22246 & 40342 & 61328 & 82914 \\
        25 & 22 & 256 & 1194 & 4310 & 12024 & 26696 & 49886 & 77746 & 107078 \\
        26 & 32 & 288 & 1398 & 4964 & 14120 & 31964 & 61042 & 97902 & 137332 \\
        27 & 22 & 304 & 1552 & 5588 & 16500 & 38248 & 73932 & 122148 & 174884 \\
        28 & 36 & 350 & 1621 & 6318 & 18926 & 45666 & 89160 & 151348 & 221106 \\
        29 & 26 & 320 & 1864 & 7144 & 21776 & 54080 & 107152 & 185784 & 277668 \\
        30 & 43 & 430 & 2044 & 8168 & 25122 & 63840 & 128784 & 226454 & 346300 \\\hline
    \end{tabular}
    \caption{The absolute degeneracies of the flavored MacDonald indices for $\U(N)$ gauge groups with $N=1, \dotsc, 9$.}
    \label{Tab:UFM}
\end{table}

\newpage 

\begin{table}[H]
    \centering
    \begin{tabular}{|c||c|c|c|c|c|c|c|c|c|}
        \hline
        $2h$ & $\SU(2)$ & $\SU(3)$ & $\SU(4)$ & $\SU(5)$ & $\SU(6)$ & $\SU(7)$ & $\SU(8)$ & $\SU(9)$ 
        \\\hline\hline
        0 & 1 & 1 & 1 & 1 & 1 & 1 & 1 & 1 \\
        1 & 0 & 0 & 0 & 0 & 0 & 0 & 0 & 0 \\
        2 & 3 & 3 & 3 & 3 & 3 & 3 & 3 & 3 \\
        3 & 4 & 4 & 4 & 4 & 4 & 4 & 4 & 4 \\
        4 & 9 & 4 & 9 & 9 & 9 & 9 & 9 & 9 \\
        5 & 12 & 8 & 10 & 16 & 16 & 16 & 16 & 16 \\
        6 & 22 & 11 & 22 & 19 & 26 & 26 & 26 & 26 \\
        7 & 36 & 16 & 26 & 36 & 40 & 48 & 48 & 48 \\
        8 & 60 & 26 & 51 & 46 & 59 & 62 & 71 & 71 \\
        9 & 88 & 32 & 66 & 64 & 96 & 112 & 118 & 128 \\
        10 & 135 & 44 & 108 & 85 & 120 & 133 & 152 & 159 \\
        11 & 204 & 68 & 148 & 116 & 208 & 220 & 266 & 288 \\
        12 & 302 & 100 & 241 & 159 & 245 & 284 & 303 & 356 \\
        13 & 432 & 128 & 326 & 180 & 404 & 400 & 554 & 580 \\
        14 & 627 & 187 & 489 & 257 & 469 & 533 & 593 & 760 \\
        15 & 900 & 252 & 676 & 308 & 780 & 688 & 1068 & 1104 \\
        16 & 1269 & 353 & 990 & 417 & 901 & 923 & 1112 & 1462 \\
        17 & 1764 & 468 & 1358 & 488 & 1436 & 1124 & 1974 & 1984 \\
        18 & 2451 & 642 & 1919 & 669 & 1664 & 1509 & 2025 & 2602 \\
        19 & 3384 & 844 & 2614 & 796 & 2600 & 1792 & 3496 & 3400 \\
        20 & 4629 & 1124 & 3660 & 1055 & 3067 & 2411 & 3466 & 4506 \\
        21 & 6268 & 1480 & 4958 & 1276 & 4620 & 2800 & 6034 & 5648 \\
        22 & 8460 & 1936 & 6759 & 1642 & 5529 & 3638 & 5803 & 7417 \\
        23 & 11376 & 2528 & 9080 & 2012 & 8124 & 4392 & 10162 & 9136 \\
        24 & 15183 & 3298 & 12288 & 2530 & 9857 & 5454 & 9762 & 11932 \\
        25 & 20124 & 4244 & 16354 & 3132 & 14040 & 6668 & 16840 & 14544 \\
        26 & 26595 & 5460 & 21789 & 3899 & 17245 & 8070 & 16243 & 18512 \\
        27 & 35008 & 7004 & 28770 & 4844 & 24152 & 9796 & 27478 & 22640 \\
        28 & 45828 & 8928 & 38043 & 5980 & 29918 & 11749 & 27027 & 28036 \\
        29 & 59700 & 11436 & 49842 & 7356 & 40920 & 14248 & 44432 & 34708 \\
        30 & 77522 & 14475 & 65161 & 9217 & 51052 & 16840 & 44471 & 41532 \\\hline
    \end{tabular}
    \caption{The absolute degeneracies of the flavored MacDonald index for $\SU(N)$ gauge groups with $N=2, \dotsc, 9$.}
    \label{Tab:SUFM}
\end{table}

\newpage
\bibliography{refs,HoloTwist}
\bibliographystyle{JHEP}

\end{document}